\documentclass[11pt]{article}
\usepackage[margin=1in]{geometry}

\usepackage{amsmath,amsthm,amssymb,verbatim,algorithm2e,mathtools}
\usepackage[noadjust]{cite}
\usepackage[usenames,dvipsnames]{xcolor}

\usepackage{algpseudocode}

\usepackage[colorlinks=true, linkcolor=red,citecolor=ForestGreen,urlcolor = black]{hyperref}

\newcommand{\remove}[1]{}
\makeatletter
\newtheorem*{rep@theorem}{\rep@title} \newcommand{\newreptheorem}[2]{%
\newenvironment{rep#1}[1]{%
\def\rep@title{\bf #2 \ref{##1} }%
\begin{rep@theorem} }%
{\end{rep@theorem} } }
\makeatother
\newreptheorem{theorem}{Theorem}
\newreptheorem{lemma}{Lemma}

\newtheorem{thm}{Theorem}[section]
\newtheorem{claim}[thm]{Claim}
\newtheorem{lemma}[thm]{Lemma}
\newtheorem{define}[thm]{Definition}

\newtheorem{THM}{Theorem}
\newtheorem{COR}[THM]{Corollary}

\renewcommand{\remove}[1]{}

\newcommand{\poly}{{\rm poly}}

\renewcommand{\l}{\left}
\renewcommand{\r}{\right}

\newcommand{\la}{{\lambda}}

\newcommand{\comments}[1]{}
\newcommand{\same}{\textnormal{$same^{\star}$}}
\newcommand{\cpy}{\textnormal{copy}}
\newcommand{\nmExt}{\textnormal{nmExt}}

\newcommand{\D}{\mathbf{\mathcal{D}}}

\newcommand{\LExt}{\textnormal{LExt}}

\newcommand{\slice}{\textnormal{Slice}}

\newcommand{\samp}{\textnormal{Samp}}
\newcommand{\IP}{\textnormal{IP}}
\newcommand{\Enc}{\textnormal{Enc}}
\newcommand{\Dec}{\textnormal{Dec}}

\newcommand{\bch}{\textnormal{BCH}}
\newcommand{\dbch}{\textnormal{dBCH}}

\newcommand{\ilnm}{\textnormal{i$\ell$NM}}
\newcommand{\ilext}{\textnormal{i$\ell$Ext}}

\newcommand{\Ext}{\textnormal{Ext}}

\newcommand{\lin}
{\textnormal{Lin}}
\newcommand{\sss}
{\textnormal{SS}}
\newcommand{\iss}
{\textnormal{ISS}}
\newcommand{\css}
{\textnormal{CSS}}

\newcommand{\con}{\textnormal{Con}}

\newcommand{\adv}{\textnormal{advGen}}

\newcommand{\ox}{\overline{\mathbf{X}}}
\newcommand{\oy}{\overline{\mathbf{Y}}}

\newcommand{\ol}[1]{\overline{#1}}
\newcommand{\bb}[1]{\mathbb{#1}}
\newcommand{\mcl}[1]{\mathcal{#1}}

\newcommand{\E} {\mathbf{E}}

\newcommand{\N}{\mathcal{N}}
\newcommand{\A}{\mathcal{A}}

\newcommand{\C}{\mathcal{C}}
\newcommand{\U}{\mathbf{U}}

\newcommand{\X}{\mathbf{X}}
\newcommand{\Y}{\mathbf{Y}}
\newcommand{\xb}{\mathbf{x}}

\newcommand{\zo}{\{0, 1\}}
\newcommand{\pr}{\mathbf{Pr}}

\newcommand{\acb}{\textnormal{ACB}}
\newcommand{\cb}{\textnormal{CB}}

\newcommand{\F}{\mathcal{F}}
\newcommand{\s}{\mathbf{S}}
\newcommand{\W}{\mathbf{W}}
\newcommand{\V}{\mathbf{V}}
\newcommand{\rr}{\mathbf{R}}

\newcommand{\Z}{\mathbf{Z}}

\newcommand{\T}{\mathbf{T}}

\makeatletter
\def\old@comma{,}
\catcode`\,=13
\def,{%
  \ifmmode%
    \old@comma\discretionary{}{}{}%
  \else%
    \old@comma%
  \fi%
}
\newcommand{\x}[1]{{}$\kern-2\mathsurround${}
  \binoppenalty10000 \relpenalty10000 #1{}$\kern-2\mathsurround${}}
  
\makeatother
\def\draft{1}   % 1 for working draft

\ifnum\draft=1 % show authors' note if draft
    \def\ShowAuthNotes{1}
\else
    \def\ShowAuthNotes{0}
\fi

\ifnum\ShowAuthNotes=1
\newcommand{\authnote}[2]{{ \footnotesize \bf{\color{red}[#1's Note: {\color{blue}#2}]}}}
\else
\newcommand{\authnote}[2]{}
\fi

\makeatletter
\def\sand{%
  \end{tabular}%
  \hskip 0.5em \@plus.17fil\relax
  \begin{tabular}[t]{c}}
\makeatother

\begin{document} 
\title{\textbf{Non-Malleable Extractors and Codes for Composition of Tampering, Interleaved Tampering and More}}
\author{ Eshan Chattopadhyay\thanks{Part of this work was done when the author was a postdoctoral researcher at the Institute for Advanced Study, Princeton, partially supported by NSF Grant CCF-1412958 and the Simons Foundation.}\\ Cornell University  \\ \href{mailto:eshanc@cornell.edu}{eshanc@cornell.edu}
 \and Xin Li\thanks{Partially supported by NSF Grant CCF-1617713.} \\ Department of Computer Science,\\ John Hopkins University \\  \href{mailto: lixints@cs.jhu.edu}{ lixints@cs.jhu.edu}}
 \maketitle
\thispagestyle{empty}

\begin{abstract}
Non-malleable codes were  introduced by  Dziembowski, Pietrzak, and Wichs (JACM 2018)  as a generalization of standard error correcting codes to handle severe forms of tampering on codewords.\ This notion has attracted a lot of recent research, resulting in various explicit constructions, which have found applications in tamper-resilient cryptography and connections to other pseudorandom objects in theoretical computer science.

 We continue the line of investigation on explicit constructions of non-malleable codes in the information theoretic setting, and give explicit constructions for several new classes of tampering functions. These classes strictly generalize several previously studied classes of tampering functions, and in particular extend the well studied split-state model which is a  ``compartmentalized" model in the sense that the codeword is partitioned \emph{a prior} into disjoint intervals for tampering. Specifically, we give explicit non-malleable codes for the following classes of tampering functions. %The following are our main results on explicit constructions of non-malleable codes.
 \begin{itemize}
 \item Interleaved split-state tampering: Here the codeword is partitioned in an unknown  way by an adversary, and then tampered with by a split-state tampering function. 
 \item Linear function composed with split-state tampering: In this model, the codeword is first tampered with by a split-state adversary, and then the whole tampered codeword is further tampered with by a linear function. In fact our results are stronger, and we can handle linear function composed with interleaved split-state tampering.
\item Bounded communication split-state tampering: In this model, the two split-state tampering adversaries are allowed to participate in a communication protocol with a bounded communication budget. 
  \end{itemize}
  Our results are the first explicit constructions of non-malleable codes in any of these tampering models. We derive all these results from explicit constructions of seedless non-malleable extractors, which we believe are of independent interest. 
  
 Using our techniques, we also give an improved seedless extractor for an unknown interleaving of two independent sources.
\end{abstract}
\clearpage 
\setcounter{page}{1}
\section{Introduction}
\subsection{Non-malleable Codes}
 Non-malleable codes were introduced by Dziembowski, Pietrzak, and Wichs \cite{DPW10} as an elegant relaxation and generalization of standard error correcting codes, where the motivation is to handle much larger classes of tampering functions on the codeword. Traditionally, error correcting codes only provide meaningful guarantees (e.g., unique decoding or list-decoding) when \emph{part} of the codeword is modified (i.e., the modified codeword is close in Hamming distance to an actual codeword), whereas in practice an adversary can possibly use much more complicated functions to modify the entire codeword. In the latter case, it is easy to see that error correction or even error detection becomes generally impossible, for example an adversary can simply change all codewords into a fixed string. On the other hand, non-malleable codes can still provide useful guarantees here, and thus partially bridge this gap. Informally, a non-malleable code guarantees that after tampering, the decoding either correctly gives the original message or gives a message that is completely unrelated and independent of the original message. This captures the notion of non-malleability: that an adversary cannot modify the codeword in a way such that the tampered codeword decodes back to a related but different message. 
 
 The original intended application of non-malleable codes is in tamper-resilient cryptography \cite{DPW10}, where they can be used generally to prevent an adversary from learning secret information by observing the input/output behavior of modified ciphertexts. Subsequently, non-malleable codes have found applications in non-malleable commitments \cite{GPR16}, public-key encryptions \cite{CMTV15}, non-malleable secret sharing schemes \cite{GK18a,GK18b}, and privacy amplification protocols \cite{CKOS18}. Furthermore, interesting connections were found to non-malleable extractors \cite{CG14b}, and very recently to spectral expanders \cite{RS18}. Along the way, the constructions of non-malleable codes used various components and sophisticated ideas from additive combinatorics \cite{ADL13,CZ14} and randomness extraction \cite{CGL15}, and some of these techniques have also found applications in constructing extractors for independent sources  \cite{Li16}. Until today, non-malleable codes have become fundamental objects at the intersection of coding theory and cryptography. They are well deserved to be studied in more depth in their own right, as well as to find more connections to other well studied objects in theoretical computer science. %This makes it particularly interesting to study non-malleable codes on its own right, and possibly find more connections to other well studied objects in theoretical computer science.
 
 We first introduce some notation before formally defining non-malleable codes. 
 \begin{define}
For any function $f:S \rightarrow S$, $f$ has a fixed point at $s \in S$ if $f(s)=s$. We say $f$ has no fixed points in  $T \subseteq S$, if $f(t) \neq t$ for all $t \in T$. $f$ has no fixed points if $f(s) \neq s$ for all $s \in S$.
\end{define}

\begin{define}[Tampering functions]For any $n>0$, let $\mathcal{F}_n$ denote the set of all functions $f: \{ 0,1\}^n \rightarrow \{0,1\}^n$. Any subset of $\mathcal{F}_n$ is a family of  tampering functions. 
\end{define}
We use the statistical distance to measure the distance between distributions.
\begin{define}
The statistical distance between two distributions $\D_1$ and $\D_2$ over some universal set $\Omega$ is  defined as $|\D_1-\D_2| = \frac{1}{2}\sum_{d \in \Omega}|\pr[\D_1=d]- \pr[\D_2=d]|$. We say $\D_1$ is $\epsilon$-close to $\D_2$ if $|\D_1-\D_2| \le \epsilon$ and denote it by $\D_1 \approx_{\epsilon} \D_2$.
\end{define}

To handle fixed points, we need to define the following function.

\[
 \cpy(x,y) =
  \begin{cases}
   x & \text{if } x \neq \same \\
   y       & \text{if } x  = \same

  \end{cases}
\]

Following the treatment in \cite{DPW10}, we first define coding schemes.
 \begin{define}[Coding schemes] Let $\Enc:\{0,1\}^k \rightarrow \{0,1\}^n$ and $\Dec:\{0,1\}^n \rightarrow \{0,1\}^k \cup \{ \perp \}$ be functions such that $\Enc$ is a randomized function (i.e.,\ it has access to private randomness) and $\Dec$ is a deterministic function. We say that $(\Enc,\Dec)$ is a coding scheme with block length~$n$ and message length $k$ if for all $s \in \{0,1\}^k $, $\Pr[\Dec(Enc(s))=s]=1$, where the probability is taken over the randomness in $\Enc$. 
\end{define}
We can now define non-malleable codes.\
%Let $\mathcal{F}_n$ denote the set of all functions mapping $n$-bit strings to $n$-bit strings.
\begin{define}[Non-malleable codes]\label{nm_def} A coding scheme $\C = (\Enc,\Dec)$  with block length $n$ and message length $k$ is a non-malleable code with respect to a family of tampering functions  $\mathcal{F} \subset \mathcal{F}_n$  and error~$\epsilon$ if  for every $f \in \mathcal{F}$ there exists a random variable $D_f$ on $\{ 0,1\}^k \cup \{ \same \}$ which is independent of the randomness in $\Enc$  such that for all messages $s \in \{0,1\}^k$, it holds that $$  |\Dec(f(\Enc(s))) - \cpy(D_f,s)| \le \epsilon. $$
We say the code is explicit if both the encoding and decoding can be done in polynomial time.\ The rate of  $\C$ is given by $k/n$. 
\end{define}

\paragraph{Relevant prior work on non-malleable codes.}
There has been a lot of exciting research on non-malleable codes, and we do not even  attempt to provide a comprehensive survey of them. Instead we focus on relevant explicit constructions in the information theoretic setting, which is also the focus of this paper. One of the most studied classes of tampering functions is the so called \emph{split-state} tampering, where the codeword is divided into (at least two) disjoint intervals and the adversary can tamper with each interval arbitrarily but independently. This model arises naturally in situations where the codeword may be stored in different parts of memory or different devices.  Following a very successful line of work \cite{DKO13,ADL13,CG14b,CZ14,ADKO14,CGL15,Li16,Li18}, we now have explicit constructions of non-malleable codes in the $2$-split state model with constant rate and constant error, or rate $\Omega(\log \log n/\log n)$ with exponentially small error \cite{Li18}.  For larger number of states, recent work of Kanukurthi, Obbattu, and Sruthi \cite{kan4}, and that of Gupta, Maji and Wang \cite{GMW17} gave explicit constructions in the $4$-split-state model and $3$-split-state model respectively, with constant rate and negligible error. %almost achieve optimal rates. The number of states required in this construction was improved to $3$ by Gupta, Maji and Wang \cite{GMW17}. 

The split state model is a ``compartmentalized" model, where the codeword is partitioned \emph{a prior} into disjoint intervals for tampering. Recently, there has been  progress towards handling non-compartmentalized tampering functions. A work of Agrawal, Gupta, Maji, Pandey and Prabhakaran \cite{AGMPP15} gave explicit constructions of non-malleable codes with respect to tampering functions that permute or flip the bits of the codeword.  Ball, Dachman-Soled, Kulkarni and Malkin \cite{BDKM16} gave explicit constructions of non-malleable codes against $t$-local functions for $t \le  n^{1-\epsilon}$. However in all these models, each bit of the tampering function only depends on part of the codeword.\ A recent work of Chattopadhyay and Li \cite{CL17} gave the first explicit constructions of non-malleable codes where each bit of the tampering function may depend on all bits of the codeword.\ Specifically, they gave constructions for the classes of linear functions and small-depth (unbounded fain-in) circuits.\ The rate of the non-malleable code with respect to small-depth circuits was exponentially improved by a subsequent work of Ball, Dachman-Soled, Guo, Malkin, and Tan \cite{Liyang18}. 

Given all these exciting results, a major goal of the research on non-malleable codes remains to give explicit constructions for broader classes of tampering functions, as one can use the probabilistic method to show the existence of non-malleable codes with rate close to $1-\delta$ for any class $\mathcal{F}$ of tampering functions with $|\mathcal{F}| \leq 2^{2^{\delta n}}$ \cite{CG14a}.

\paragraph{Our results.}
We continue the line of investigation on explicit constructions of non-malleable codes, and give explicit constructions for several new classes of non-compartmentalized tampering functions, where in some classes each bit of the tampering function can depend on all the bits of the codeword. The new classes strictly generalize several previous studied classes of tampering functions. In particular, we consider the following three classes. 
\begin{enumerate}
\item \emph{Interleaved $2$-split-state tampering}, where the adversary can divide the codeword into two arbitrary disjoint intervals and tamper with each interval arbitrarily but independently. This model generalizes the split-state model and captures the situation where the codeword is partitioned into two halves in an unknown way by the adversary before applying a $2$-split-state tampering function. Constructing non-malleable codes for this class of tampering functions was left as an open problem  by Cheraghchi and Guruswami \cite{CG14b}. \item \emph{Composition of tampering}, where the adversary takes two tampering functions and compose them together to get a new tampering function.\ We note that function composition is a natural strategy for an adversary to achieve more powerful tampering, and it has been studied widely in other fields (e.g., computational complexity and communication complexity).\ Thus we believe that studying non-malleable codes for the composition of known classes of tampering functions is also a natural and important direction. \item \emph{Bounded communication $2$-split-state tampering}, where the two tampering functions in a $2$-split state model are allowed to have some bounded communication. 
\end{enumerate}
We now formally define these classes and some related classes below. We use the notation that for any permutation $\pi:[n] \rightarrow [n]$ and any string $x \in [r]^{n}$, $y=x_{\pi}$ denotes the length $n$ string such that $y_{\pi(i)}=x_{i}$.%We note that%dividedagainst tampering functions that have global access to the codeword. As discussed above, recent progress makes assumption on the complexity of the tampering class when handling such global tampering functions (e.g., low-depth circuits). Our approach towards handling global tampering is to look at  non-trivial extensions of the $2\sss$ model of tampering. In particular the classes of functions we handle are  $2\iss, \lin \circ 2\sss$, and $(2,t,\ell)-\css$. As discussed above, $2\iss$ models an adversary that decides a  partition of the codeword before applying a split-state adversary. $\lin \circ 2\sss$ initiates study towards handling tampering of a bounded complexity class function on the top of split-state tampering. Finally,  $(2,t,\ell)-\css$ models the natural extension when the split-state adversary participates in a communication protocol before fixing the split-state tampering functions.

\begin{itemize}
\item The family of $2$-split-state functions $2\sss \subset \F_{2n}$: Any $f \in 2\sss$ comprises of two functions $f_1: \zo^n \rightarrow \zo^n$ and $f_2: \zo^n \rightarrow \zo^n$, and for any $x, y \in \zo^n$, $f(x,y) = (f_1(x), f_2(x))$. This family of tampering functions has been extensively studied, with a long line of work achieving near optimal explicit constructions of non-malleable codes.
\item The family of linear functions $\lin \subset \F_{n}$: Any $f \in \lin$ is a linear function from $\zo^n$ to $\zo^n$ (viewing $\zo^n$ as $\mathbb{F}_2^n)$. 
\item The family of interleaved $2$-split-state functions $2\iss \subset \F_{2n}$:  Any $f \in 2\iss$ comprises of two functions $f_1: \zo^n \rightarrow \zo^n$, $f_2: \zo^n \rightarrow \zo^n$, and a permutation $\pi: [2n] \rightarrow [2n]$. For any $z  = (x \circ y)_{\pi} \in \zo^{2n}$,  where $x,y \in \zo^n$,  let $f(z) = (f_1(x) \circ f_2(y))_{\pi}$ (where $\circ$ denotes the string concatenation operation). %This model generalizes the split-state model and captures the situation where the codeword is partitioned into two halves by the adversary before applying a $2$-split-state tampering function. Constructing non-malleable codes for this class of tampering functions was left as an open problem  by Cheraghchi and Guruswami \cite{CG14b}, and no explicit construction was known prior to this work.
\item The family of bounded communication $2$-split-state functions $(2,t)-\css$:\ Consider the following natural extension of the $2$-split-state model.\ Let $c= (x,y)$ be a codeword in $\zo^{2n}$, where $x$ is the first $n$ bits of $c$ and $y$ is the remaining $n$ bits of $c$. Let Alice and Bob be two tampering adversaries, where Alice has access to $x$ and Bob has access to $y$. Alice and Bob run a (deterministic) communication protocol based on $x$ and $y$ respectively, which can last for an arbitrary number of rounds but each party sends at most $t$ bits in total. Finally, based on the transcript and $x$ Alice outputs $x' \in \zo^n$, similarly based on the transcript and $y$ Bob outputs $y' \in \zo^n$. The tampered codeword is $c'=(x',y')$.  
\item For any tampering function families $\F, \mathcal{G} \subset \F_{n}$, define the family $\F \circ \mathcal{G} \subset \F_n$ to be the set of all functions  of the form $f \circ g$, where $f \in \F$, $g \in \mathcal{G}$ and $\circ$ denotes function composition. %In this paper, we consider the case where $\F=\lin$ and $\mathcal{G}=2\sss$, i.e., linear function composed with $2$-split-state tampering.
\end{itemize} 

 We now formally state  our results.\ Our main result is an explicit non-malleable code with respect to the tampering class of $\lin \circ 2\iss$, i.e, linear function composed with interleaved $2$-split-state tampering. Specifically, we have the following theorem. 
 \begin{THM}
 There exists a constant $\delta>0$ such that for all integers $n>0$ there exists an explicit non-malleable code with respect to $\lin \circ 2\iss$ with rate $1/n^{\delta}$ and error $2^{-n^{\delta}}$.
 \end{THM}
 
This immediately gives the following corollaries, which give explicit non-malleable codes for interleaved $2$-split-state tampering, and linear function composed with $2$-split-state tampering.
\begin{COR}
\label{cor:intro_int_nm_code}
There exists a constant $\delta>0$ such that for all integers $n>0$ there exists an explicit non-malleable code with respect to $2\iss$ with rate $1/n^{\delta}$ and error $2^{-n^{\delta}}$.
\end{COR}
\begin{COR}
\label{cor:intro_comp_nm_code}
There exists a constant $\delta>0$ such that for all integers $n>0$ there exists an explicit non-malleable code with respect to $\lin \circ 2\sss$ with rate $1/n^{\delta}$ and error $2^{-n^{\delta}}$.
\end{COR}

Next we give an explicit non-malleable code with respect to bounded communication $2$-split-state tampering.
\begin{THM}
\label{thm:intro_com_nm_code}
There exists a constant $\delta>0$ such that for all integers $n, t>0$ with $t \le \delta n$, there exists an explicit non-malleable code with respect to $(2,t)-\css$ with rate $\Omega(\log \log n/\log n)$ and error $2^{-\Omega(n \log \log n/\log n)}$.
\end{THM}
\noindent Prior to our work, no explicit non-malleable code of any rate was known for these tampering classes. %$   2\iss, \lin \circ 2\sss$, or $(2,t,\ell)-\css$. 

\subsection{Seedless non-malleable extractors}
Our results on non-malleable codes are based on new constructions of seedless non-malleable extractors, which we believe are of independent interest. Before defining seedless non-malleable extractors formally, we first recall some basic notation from the area of randomness extraction.

Randomness extraction is motivated by the problem of purifying imperfect (or defective) sources of randomness. The concern stems from the fact that natural random sources often have poor quality, while  most applications require high quality (e.g., uniform) random bits.\ We use the standard notion of min-entropy to measure the amount of randomness in a distribution.

\begin{define} The min-entropy $H_{\infty}(\X)$ of a probability distribution $\X$ is defined to be \newline $\min_{x}(-\log(\Pr[\X=x]))$. We say a probability distribution $\X$ on $\{ 0,1\}^n$ is an $(n, H_{\infty}(\X))$-source and the min-entropy rate is $H_{\infty}(\X)/n$. %Any source $\X$ on $\{ 0,1\}^n$ with min-entropy at least $k$ is called an $(n,k)$-source. 
\end{define}
It turns out that it is impossible to extract from a single general weak random source even for min-entropy $n-1$. There are two possible ways to bypass this barrier. The first one is to relax the extractor to be a \emph{seeded extractor}, which takes an additional independent short random seed to extract from a weak random source. The second one is to construct deterministic extractors for special classes of weak random sources. %a particularly useful and well studied notion is that of a seeded extractor, which takes an additional  short random seed to extract from a weak random source.
  
Both kinds of extractors have been studied extensively. Recently, they have also been generalized to stronger notions where the inputs to the extractor can be tampered with by an adversary. Specifically, Dodis and Wichs \cite{DW09} introduced the notion of \emph{seeded non-malleable extractor} in the context of privacy amplification against an active adversary.\ Informally, such an extractor satisfies the stronger property that the output of the extractor is independent of the output of the extractor on a tampered seed. Similarly, and more relevant to this paper, a seedless variant of non-malleable extractors was introduced by Cheraghchi and Guruswami \cite{CG14b} as a way to construct non-malleable codes. Apart from their original applications, both kinds of non-malleable extractors are of independent interest. They are also related to each other and have applications in constructions of extractors for independent sources \cite{Li16}. 

We now define seedless non-malleable extractors. For simplicity, the definition here assumes that the tampering function has no fixed points. See Section $\ref{section:prelims}$ for a more formal definition. 
\begin{define}[Seedless non-malleable extractors] Let $\F \subset \F_n$ be a family of tampering functions such that no function in $\F$ has any fixed points.  A function $\nmExt : \{0,1\}^{n} \rightarrow \{0,1\}^{m}$ is  a seedless $(n,m,\epsilon)$-non-malleable extractor with respect to $\F$ and  a class of sources $\mathcal{X}$   if for every distribution $\X \in \mathcal{X}$ and every tampering function $f \in \F$,   $$|\nmExt(\X), \nmExt(f(\X)) - \U_m, \nmExt(f(\X)) | \le \epsilon. $$ Further, we say that $\nmExt$ is $\epsilon'$-invertible, if there exists a polynomial time sampling algorithm $\A$ that takes as input $y \in \zo^m$, and outputs a  sample from a distribution that is $\epsilon'$-close to the uniform distribution on the set $\nmExt^{-1}(y)$.
\end{define}
In the above definition, when the class of sources $\mathcal{X}$ is the distribution $\U_n$,  we simply say that $\nmExt$ is a   seedless $(n,m,\epsilon)$-non-malleable extractor with respect to $\F$.

\iffalse
\begin{define}[$(r,n,k)$-Interleaved Sources] Let $\X_1,\ldots,\X_r$ be independent $(n,k)$-sources and let $\pi:[rn] \rightarrow [rn]$ be any permutation. Then,  $Z=(\X_1\circ \ldots \circ \X_r)_{\pi}$ is an $(r,n,k)$-interleaved source.
\end{define}
If the parameters $n$ and $k$ are clear from the context, we sometimes drop them and talk about $r$-interleaved sources.
We now specialize the definition of seedless non-malleable extractors to the case of interleaved tampering function family.

\begin{define} A function $\nmExt: \zo^{r n}\rightarrow \zo^{m}$ is a $(r,n,k,\epsilon)$-interleaved non-malleable extractors if for any $(r,n,k)$-interleaved source $\Z$ and any $(r,n)$-interleaved tampering function $f:\zo^{rn} \rightarrow \zo^{rn}$, $$ \nmExt(\Z), \nmExt(f(\Z)) \approx_{\epsilon} \U_m, \nmExt(f(\Z)).$$
\end{define}
\fi

\paragraph{Relevant prior work on seedless non-malleable extractors.}
The first construction of seedless non-malleable extractors was given by Chattopadhyay and Zuckerman \cite{CZ14} with respect to the class of $10$-split-state tampering. Subsequently, a series of works starting with the work of Chattopadhyay,  Goyal and Li \cite{CGL15} gave explicit seedless non-malleable extractors for $2$-split-state tampering. The only known construction with respect to a class of tampering functions different from split state tampering is the work of Chattopadhyay and Li \cite{CL17}, which gave explicit seedless non-malleable extractors with respect to the tampering class $\lin$ and small depth circuits. We note that constructing explicit seedless non-malleable extractors with respect to $2\iss$ was also posed as an open problem in \cite{CG14b}.

\paragraph{Our results.}
We give the first explicit constructions of seedless non-malleable extractors with respect to the tampering classes $\lin \circ 2\iss$ and $(2,t)-\css$.\ Note that the first construction also directly implies non-malleable extractors with respect to the classes $2\iss$ and $\lin \circ 2\sss$.\ The non-malleable extractors with respect to $\lin \circ 2\iss$ is a  fundamentally new construction. The non-malleable extractor with respect to $(2,t)-\css$ is obtained by showing a reduction to seedless non-malleable extractors for $2\sss$, where excellent constructions are known (e.g., a recent construction of Li \cite{Li18}).

We now formally state our main results.
\begin{THM}
For all $n>0$ there exists an efficiently computable seedless $(n,n^{\Omega(1)},2^{-n^{\Omega(1)}})$-non-malleable extractor with respect to $\lin \circ 2\iss$,  that is $2^{-n^{\Omega(1)}}$-invertible.
\end{THM}

This immediately gives the following two corollaries.
 \begin{COR}
\label{cor:intro_int_ext}
For all $n>0$ there exists an efficiently computable seedless $(n,n^{\Omega(1)},2^{-n^{\Omega(1)}})$-non-malleable extractor with respect to $2\iss$, that is $2^{-n^{\Omega(1)}}$-invertible.
\end{COR}
\begin{COR}
\label{cor:intro_comp_ext}
For all $n>0$ there exists an efficiently computable seedless $(n,n^{\Omega(1)},2^{-n^{\Omega(1)}})$-non-malleable extractor with respect to $\lin \circ 2\sss$,  that is $2^{-n^{\Omega(1)}}$-invertible.
\end{COR}

Next we give the non-malleable extractor with respect to $(2,t)-\css$.
\begin{THM}
\label{thm:intro_com_ext}
There exists a constant $\delta>0$ such for all integers $n,t>0$ with $t \le \delta n$, there exists an efficiently computable seedless $(n,\Omega\l(\frac{ n \cdot (\log \log n) }{\log n}\r),2^{-\Omega(n \log \log n/\log n)})$-non-malleable extractor with respect to $(2,t)-\css$, that is $2^{-\Omega(n \log \log n/\log n)}$-invertible.
\end{THM}
We derive our results on non-malleable codes using the above explicit constructions of non-malleable extractors. In particular we use the following theorem  proved by Cheraghchi and Guruswami \cite{CG14b} that connects non-malleable extractors and codes.
\begin{thm}[\cite{CG14b}]\label{thm:connection} Let $\nmExt: \{0,1\}^{n} \rightarrow \{0,1\}^{m}$  be an efficient seedless $(n,m,\epsilon)$-non-malleable extractor with respect to a class of tampering functions $\F$ acting on $\zo^n$. Further suppose $\nmExt$ is $\epsilon'$-invertible.

Then there exists an efficient construction of a non-malleable code  with respect to the tampering family $\F$ with block length $=n$, relative rate  $\frac{m}{n}$ and error $2^{m}\epsilon+\epsilon'$.
\end{thm}

\subsection{Extractors for interleaved sources}\label{intro_sec_ilext}
Our techniques also yield  improved explicit constructions of  extractors for interleaved sources, which generalize extractors for independent sources in the following way: the inputs to the extractor are samples from a few independent sources mixed (interleaved) in an unknown (but fixed) way. Raz and Yehudayoff \cite{RY08} showed that such extractors have applications in communication complexity and proving lower bounds for arithmetic circuits. In a subsequent work, Chattopadhyay and Zuckerman \cite{CZ15a} showed that such extractors can also be used to construct extractors for certain samplable sources, extending a line of work initiated by Trevisan and Vadhan \cite{TV00}.  We now define interleaved sources formally.
\begin{define}[Interleaved Sources] Let $\X_1,\ldots,\X_r$ be arbitrary independent sources on $\zo^n$  and let $\pi : [rn] \rightarrow [rn] $ be any permutation. Then $Z = (\X_1 \circ  \ldots \circ \X_r)_{\pi}$ is an $r$-interleaved source.
\end{define}
  
\paragraph{Relevant prior work on interleaved extractors.} Raz and Yehudayoff \cite{RY08} gave explicit extractors  for $2$-interleaved sources when both the sources have min-entropy at least $(1-\delta)n$ for a tiny constant $\delta>0$. Their construction is based on techniques from additive combinatorics and can output $\Omega(n)$ bits with exponentially small error. Subsequently, Chattopadhyay and Zuckerman \cite{CZ15a} constructed extractors for $2$-interleaved sources where one source has entropy $(1-\gamma)n$ for a small constant $\gamma>0$ and the other source has entropy $\Omega(\log n)$.\ They achieve output length $O(\log n)$ bits with error $n^{-\Omega(1)}$. 

A much better result (in terms of the min-entropy) is known if the extractor has access to an interleaving of more sources. For a large enough constant $C$, Chattopadhyay and Li \cite{CL15} gave an explicit extractor for $C$-interleaved sources where each source has entropy $k \ge \poly(\log n)$. They achieve output length $k^{\Omega(1)}$ and error $n^{-\Omega(1)}$.

\paragraph{Our results.}
Our main result is an explicit extractor for $2$-interleaved sources where each source has min-entropy at least $2n/3$. The extractor outputs $\Omega(n)$ bits with error $2^{-n^{\Omega(1)}}$. %More formally, we have the following result.
\begin{THM}
\label{thm:intro_il_ext}
For any constant $\delta>0$ and  all integers $n>0$, there exists an efficiently computable function $\ilext: \zo^{2n} \rightarrow \zo^{m}$, $m = \Omega(n)$, such that for any two independent sources $\X$ and $\Y$, each on $n$ bits with min-entropy at least $(2/3 + \delta)n$, and any permutation  $\pi:[2n] \rightarrow [2n]$, $$|\ilext((\X \circ \Y)_{\pi}) - \U_m| \le 2^{-n^{\Omega(1)}}.$$
\end{THM}
\subsection{Open questions}
\paragraph{Non-malleable codes for composition of  function classes} We gave efficient constructions of non-malleable codes for the tampering class  $\lin \circ 2\sss $ (and more generally $\lin \circ \iss$). Many natural questions remain to be answered. For instance, one open problem is to efficiently construct non-malleable codes for the tampering class $2\sss \circ \lin$. It looks like one needs substantially new ideas to give such constructions. More generally, for what other interesting classes of functins $\F$ and $\mcl{G}$ can we construct non-malleable codes for the composed class $\F \circ \mcl{G}$? Is it possible to efficiently construct non-malleable codes for the tampering class $\F \circ \mcl{G}$ if we have efficient non-malleable codes for the classes $\F$ and $\mcl{G}$?

\paragraph{Other applications for seedless non-malleable extractors} The explicit seedless non-malleable extractors that we construct satisfy strong pseudorandom properties. A natural question is to find more applications of these non-malleable extractors in explicit constructions of other interesting  objects.

\paragraph{Improved seedless extractors} We construct an extractor for $2$-interleaved sources that works for min-entropy rate $2/3$. It is easy to verify that there exists extractors for sources with min-entropy as low as $ C \log n$, and a natural question here is to come up with such explicit constructions. Given the success in constructing $2$-source extractors for low min-entropy \cite{CZ15,Li18}, we are hopeful that more progress can be made on this problem.

\subsection{Organization}
The rest of the paper is organized as follows. We use Section $\ref{sec:overview}$ to present an overview of our results and techniques. We use Section~$\ref{section:prelims}$ to introduce some background and notation.  We present our seedless non-malleable extractor construction with respect to $\lin \circ 2\iss$  in Section~$\ref{section:composed_split_state}$. We use Section $\ref{section:communication_split_state}$ to present our non-malleable extractor construction with respect to $(2,t)-\css$.  We present efficient sampling algorithms for  our seedless non-malleable extractor constructions in Section $\ref{sec:sampling}$. We use Section $\ref{sec:ilext}$ to present an explicit construction of an extractor for interleaved sources. 
\section{Overview of constructions and techniques}
\label{sec:overview}

Our results on non-malleable codes are derived from explicit constructions of invertible seedless non-malleable extractors (see Theorem $\ref{thm:connection}$). In this section, we focus on explicit constructions of seedless non-malleable extractors with respect to the relevant classes of tampering functions, and  explicit extractors for interleaved sources. %We conclude by discussing the explicit extractor for interleaved sources.

\paragraph{Seedless non-malleable extractors with respect to $\lin \circ 2\iss$.}
 We  construct a seedless non-malleable extractor $\nmExt:\zo^{n} \times \zo^n \rightarrow \zo^{m}$, $m=n^{\Omega(1)}$ such that the following hold: Let $\X$ and $\Y$ be two independent uniform sources, each on $n$ bits. Let $h:\zo^{2n} \rightarrow \zo^{2n}$ be an arbitrary linear function,  $f:\zo^n \rightarrow \zo^{n}$, $g:\zo^{n} \rightarrow \zo^n$ be two arbitrary functions, and $\pi:[2n] \rightarrow [2n]$ be an arbitrary permutation. Then, $$ \nmExt((\X,\Y)_{\pi}), \nmExt(h((f(\X) \circ g(\Y))_{\pi})) \approx_{\epsilon} \U_m, \nmExt(h((f(\X) \circ g(\Y))_{\pi})),$$ where $\epsilon=2^{-n^{\Omega(1)}}$. Notice that such an extractor is not possible to construct in general, for example when all $f, g, h$ are the identify function. However,  such an extractor exists when the composed function does not have fixed points. For simplicity, we ignore this issue related to fixed points in the proof sketch, and just mention that we have a reduction from the problem of constructing non-malleable codes to the problem of constructing a non-malleable extractor with no fixed points. The argument is similar to the argument in \cite{CG14b} but more complicated since here we are dealing with more powerful adversaries. We refer the reader to Section $\ref{section:composed_split_state}$ for more details.

 Our first step is to reduce the problem of constructing non-malleable codes with respect to $\lin \circ 2\iss$ to constructing non-malleable extractors with the following guarantee.   For strings $x, y \in \zo^{n}$, we use $x+y$ (or equivalently $x-y$) to denote the bit-wise xor of the two strings. Let $\X$ and $\Y$ to be two independent $(n,n-n^{\delta})$-sources and $f_1,f_2,g_1,g_2 \in \F_n$ be four functions that satisfy the following condition: 
\begin{itemize}
\item $\forall x \in support(\X)$ and $y \in support(\Y)$, $f_1(x) + g_1(y) \neq x$ or 
\item $\forall x \in support(\X)$ and $ y \in support(\Y)$, $f_2(x) + g_2(y) \neq y$.
\end{itemize} Then, 
\begin{align*}
|\nmExt((\X,\Y)_{\pi}),  \nmExt((f_1(\X) + g_1(\Y),f_2(\X) + g_2(\Y))_{\pi}) -  \\ \U_{m},  \nmExt((f_1(\X) + g_1(\Y),f_2(\X) + g_2(\Y))_{\pi})| &\le 2^{-n^{\Omega(1)}}.
\end{align*}

The reduction can be seen in the following way: Define $\ol{x}=(x, 0^n)_{\pi}$ and $\ol{y}=(0^n, y)_{\pi}$. Similarly define $\ol{f(x)}= h((f(x), 0^n)_{\pi})$ and $\ol{g(y)} = h((0^n, g(y))_{\pi})$.\ Thus, $(x, y)_{\pi} = \ol{x} + \ol{y}$ and $h((f(x), g(y))_{\pi}) = \ol{f(x)} + \ol{g(y)}$. Define functions $h_1:\zo^{2n} \rightarrow \zo^n$ and    $h_2:\zo^{2n} \rightarrow \zo^n$ such that $h((f(x) , g(y))_{\pi})= (h_1(x,y) , h_2(x,y))_{\pi}$. Since $h((f(x), g(y))_{\pi}) = \ol{f(x)} + \ol{g(y)}$, it follows that there exists functions $f_1,g_1, f_2, g_2 \in \F_n$ such that for all $x,y \in \zo^n$, the following hold: 
 
\begin{itemize}
\item $h_1(x,y) = f_1(x) + g_1(y)$, and
\item $h_2(x,y) = f_2(x) + g_2(y)$.
\end{itemize}
Thus, $h((f(x) , g(y))_{\pi}) = ((f_1(x) + g_1(y)) , (f_2(x)+ g_2(y)))_{\pi}$. 
The loss of entropy in $\X$ and $\Y$ in the reduction (from $n$ to $n-n^{\delta}$) is due to the fact that we have to handle issues related to fixed points of the tampering functions, and we ignore it for the proof sketch here. 

The idea now is to use the framework of advice generators and correlation breakers with advice to construct the non-malleable extractor \cite{C15,CGL15}. We informally define these objects below as we describe our explicit constructions.

We start with the construction of the advice generator $\adv:\zo^{2n} \rightarrow \zo^{a}$. Informally, $\adv$ is a weaker object than a non-malleable extractor, and we only need that $\adv((\X , \Y)_{\pi}) \neq \adv((f_1(\X) + g_1(\Y),f_2(\X) + g_2(\Y))_{\pi})$ (with high probability). Further, it is crucial that $a \ll n$, and in particular think of $a=n^{\gamma}$ for a small constant $\gamma>0$.  
Without loss of generality, suppose that $\forall x \in support(\X)$ and $y \in support(\Y)$, $f_1(x) + g_1(y) \neq x$.

Let $\Z=(\X , \Y)_{\pi}$.  Let $n_0=n^{\delta'}$ for some small constant $\delta'>0$. We take two slices from $\Z$, say $\Z_1$ and $\Z_2$ of lengths $n_1 =  n_0^{c_0}$ and $n_2= 10 n_0$, for some constant $c_0>1$.  Next, we use a good linear error correcting code (let the encoder of this code be E) to encode $\Z$ and sample $n^{\gamma}$ coordinates (let $\s$ denote this set)  from this encoding using $\Z_1$ (the sampler is based on seeded extractors \cite{Z97}). Let $\W_1 = E(\Z)_{\s}$.  Next, using $\Z_2$, we sample a random set of indices $\T \subset [2n]$, and let $\Z_3=\Z_{\T}$. We now use an extractor for interleaved sources, i.e., an extractor that takes as input an unknown interleaving of two independent sources and outputs uniform bits (see Section $\ref{intro_sec_ilext}$). Let $\ilext$ be this extractor (say from Theorem $\ref{thm:intro_il_ext}$), and we apply it to $\Z_3$ to get $\rr = \ilext(\Z_3)$. Finally, let  $\W_2$ be the output of a linear seeded extractor\footnote{A linear seeded extractor is a seeded extractor where for any fixing of the seed, the output is a linear function of the source.} $\LExt$ on $\Z$ with $\rr$ as the seed. The output of the advice generator is $\Z_1 \circ \Z_2 \circ \Z_3 \circ \W_1 \circ \W_2$. 

The intuition that this works is as follows.  We use the notation that if $\W = h((\X , \Y)_{\pi})$ (for some function $h$), then $\W'$ or $(\W)'$ stands for the corresponding random variable after tampering, i.e., $h(((f_1(\X) + g_1(\Y)) , (f_2(\X) + g_2(\Y)))_{\pi})$. Further, let $\X_i$ be the bits of $\X$ in $\Z_i$ for $i=1,2,3$ and $\X_4$ be the remaining bits of $\X$. Similarly define $\Y_i$'s, $i=1,2,3,4$. 
Without loss of generality suppose that $|\X_1| \ge |\Y_1|$, (where $|\alpha|$ denotes the length of the string $\alpha$).

The correctness of $\adv$ is direct if  $\Z_i \neq \Z_i'$ for some $i \in\{1,2,3\}$. Thus, assume $\Z_i = \Z_i'$ for $i=1,2,3$. It follows that hence $\s=\s',\T = \T'$ and $\rr = \rr'$. Recall that $(\X, \Y)_{\pi} = \ol{\X} + \ol{\Y}$ and $h((f(\X), g(\Y))_{\pi}) = \ol{f(\X)} + \ol{g(\Y)}$. Since $E$ is a linear code and $\LExt$ is a linear seeded extractor, the following hold: 
\begin{align*}
\W_{1}  - \W_{1}' = (E(\ol{\X} + \ol{\Y} - \ol{f(\X)} - \ol{g(\Y)} ))_{\s}, \\
\W_{2}  - \W_{2}' = \LExt(\ol{\X} + \ol{\Y} - \ol{f(\X)} - \ol{g(\Y)},\rr ). 
\end{align*}

 Now the idea is the following: Either (i) we can fix $\ol{\X}-\ol{f(\X)}$ and claim that $\X_1$ still has enough min-entropy, or (ii) we can claim that $\ol{\X}-\ol{f(\X)}$ has enough min-entropy conditioned on the fixing of $(\X_2, \X_3)$. Let us first discuss why this is enough. Suppose we are in the first case. Then, we can fix $\ol{\X}-\ol{f(\X)}$ and $\Y$ and argue that $\Z_1$ is a deterministic function of $\X$ and contains enough entropy. Note that $\ol{\X} + \ol{\Y} - \ol{f(\X)} - \ol{g(\Y)}$ is now fixed, and in fact it is fixed to a non-zero string (using the assumption that $f_1(x) + g_1(y) \neq x$). Thus, $E(\ol{\X} + \ol{\Y} - \ol{f(\X)} - \ol{g(\Y)} )$ is a string with a constant fraction of the coordinates set to $1$ (since $E$ is an encoder of a linear error correcting code with constant relative distance), and it follows that with high probability $(E(\ol{\X} + \ol{\Y} - \ol{f(\X)} - \ol{g(\Y)} ))_{\s}$ contains a non-zero entry (using the fact that $\s$ is sampled using $\Z_1$, which has enough entropy). This finishes the proof in this case since it implies $\W_{1}  \neq \W_{1}' $ with high probability.
 
Now suppose we are in case $(ii)$. We use the fact that $\Z_2$ contains entropy to conclude that the sampled bits $\Z_3$ contain almost equal number of bits from $\X$  and $\Y$ (with high probability over $\Z_2$).  Now we can fix $\Z_2$ without loosing too much entropy from $\Z_3$ (by making the size of $\Z_3$ to be significantly larger than $\Z_2$). Next, we observe that $\Z_3$ is an interleaved source, and hence $\rr$ is close to uniform. We now fix $\X_3$, and argue that $\rr$ continues to be uniform. This follows roughly from  the fact that any $2$-source extractor is strong \cite{rao2007exposition}, which easily extends to extractors for $2$ interleaved sources. Thus, $\rr$ now becomes a deterministic function of $\Y$ while at the same time, $\ol{\X}-\ol{f(\X)}$ still has enough min-entropy. Hence, $\LExt(\ol{\X}- \ol{f(\X)},\rr )$ is close to uniform even conditioned on $\rr$. We can now fix $\rr$ and  $\LExt(\ol{\Y} - \ol{g(\Y)},\rr )$  without affecting the distribution $\LExt(\ol{\X}- \ol{f(\X)},\rr )$, since $\LExt(\ol{\Y} - \ol{g(\Y)},\rr )$ is a deterministic function of $\Y$ while $\LExt(\ol{\X}- \ol{f(\X)},\rr )$ is a deterministic function of $\X$ conditioned on the previous fixing of $\rr$. It follows that after these fixings, $\W_{2}  - \W_{2}' $ is close to a uniform string  and hence $\W_{2}  - \W_{2}' \neq 0$ with probability $1- 2^{-n^{\Omega(1)}}$, which completes the proof. 
 
 The fact that we can only consider case $(i)$ and case $(ii)$  relies on a careful convex combination argument, which is in turn based on the pre-image size of the function $\tau:\zo^{n} \rightarrow \zo^{2n}$ defined as $\tau(x) = (x , 0^n)_{\pi} - h((f(x) , 0^n)_{\pi})=\ol{x}-\ol{f(x)}$. The intuition is as follows. If conditioned on the fixing of $\tau(\X)=\ol{\X}-\ol{f(\X)}$ we have that $\X$ still has very high min-entropy, then we can take the slice $\Z_1$ such that $\X_1$ still has enough entropy conditioned on the fixing of $\tau(\X)$. On the other hand, if conditioned on the fixing of $\tau(\X)$ we have that $\X$ does not have high min-entropy, then $\tau(\X)$ itself must have a large support size (or relatively high entropy). Therefore we can take the slice $\Z_2$ and sample a short string $\Z_3$ such that conditioned on the fixing of $(\X_2, \X_3)$,  $\tau(\X)$ still has enough min-entropy. To make the whole argument work, we need to carefully choose the sizes of the three slices $\Z_1, \Z_2, \Z_3$. In particular, we need to ensure that the size of $(\Z_2, \Z_3)$ is much smaller than that of $\Z_1$.

We now discuss the other crucial component in the construction, the advice correlation breaker $\acb:\zo^{2n} \times \zo^{a} \rightarrow  \zo^{m}$. Informally, $\acb$ takes $2$ inputs, a source $\Z$  (that contains some min-entropy) and an advice string $s \in \zo^a$, and outputs a distribution on $\zo^m$ with the following guarantee. If $\Z'$ is the distribution of $\Z$ after tampering, and $s' \in \zo^a$ is another advice such that $s\neq s'$, then $\acb(\Z,s), \acb(\Z',s') \approx \U_m, \acb(\Z',s')$. Typically, we also assume some structures in $\Z$ (e.g., it consists of two independent sources or an interleaving of two independent sources). Our main result is an advice correlation breaker that satisfies
\begin{align*}
\acb(\overline{\X} + \ol{\Y},s), \acb(\ol{f(\X)} + \ol{g(\Y)},s') \approx_{\epsilon} \U_m, \acb(\ol{f(\X)} + \ol{g(\Y)},s'),
\end{align*}
for any fixed strings $s,s' \in \zo^a$ where $s \neq s'$. We note that correlation breakers have found important applications in explicit constructions of seedless extractors \cite{C15,li2016improved,Li18}, thus we believe the above  correlation breaker can be of independent interest and potentially find other applications.\ By composing the advice generator $\adv$ and the correlation breaker $\acb$ in the natural way, we get the non-malleable extractor.\ Here we only briefly mention that the above advice correlation breaker crucially exploits the ``sum-structure" of the source and the tampering function, the fact that extractors are samplers \cite{Z97}, and previous constructions of correlation breakers using linear seeded extractors \cite{CL17}. We refer the reader to Section $\ref{sec:new_ext_composed}$ for more details. 
 
Finally, it is far from obvious how to efficiently invert the extractor, or in other words, sample from the pre-image of this non-malleable extractor.\ This is important since the encoder of the corresponding non-malleable code is doing exactly the sampling and thus we need it to be efficient. We use Section \ref{sec:sampling} to suitably modify our extractor to support efficient sampling.  Here we briefly sketch some high level ideas involved. Recall $\Z=(\X \circ \Y)_{\pi}$. The first modification is that in all applications of seeded extractors in our construction, we specifically use linear seeded extractors. This allows us to argue that the pre-image we are trying to sample from is in fact a convex combination of distributions supported on subspaces. The next crucial observation is the fact that we can use smaller disjoint slices of $\Z$ to carry out various steps outlined in the construction. This is to ensure that the dimensions of the subspaces that we need to sample from,  do not depend on the values of random variables that we fix. For the steps where we use the entire source $\Z$ (in  the construction of the advice correlation breaker), we replace $\Z$ by a large enough slice of $\Z$. However this is problematic if we choose the slice deterministically, since in an arbitrary interleaving of two sources, a slice of length less than $n$ might have bits only from one source. We get around this by  pseudorandomly sampling enough coordinates from $\Z$ (by first taking a small slice of $\Z$ and using a sampler that works for weak sources \cite{Z97}).

We now use an elegant trick introduced by Li \cite{Li16} where the output of the non-malleable extractor described above (with the modifications that we have specified) is now used as a seed to a linear seeded extractor applied on an even larger pseudorandom slice of $\Z$. The linear seeded extractor that we use has the property that for any fixing of the seed, the rank of the linear map corresponding to the extractor is the same, and furthermore one can efficiently sample from the pre-image of any output of the extractor.\ The final modification needed is a careful choice of the error correcting code used in the advice generator.\ For this we use a dual BCH code, which allows us to argue that we can discard some output bits of the advice generator without affecting its correctness (based on the dual distance of the code). This is crucial in order to argue that  the rank of the linear restriction imposed on the free variables of $\Z$ does not depend on the values of the bits fixed so far.  We refer the reader to Section $\ref{sec:sampling}$  for more details.

 \paragraph{Non-malleable extractors for $(2,t)-\css$.}
 We show that any $2$-source non-malleable extractor that works for min-entropy $n-2\delta n$ can be used as non-malleable extractor with respect to $(2,t)-\css$ for $t \le \delta n$. The  tampering function $h$ that is based on the communication protocol can be rephrased in terms of functions in the following way. Suppose the protocol lasts for $\ell$ rounds, there exist deterministic functions  $f_{i}$ and $g_{i}$ for $i =1,\ldots,\ell$, and $f:\zo^{n} \times \zo^{2 t} \rightarrow \zo^n$ and $g:\zo^{n} \times \zo^{2 t} \rightarrow \zo^n$ such that the  communication protocol between Alice and Bob corresponds to computing the following random variables: $\s_1 = f_1(\X), \rr_1 = g_1(\Y,\s_1),\s_2 = f_2(\X,\s_1,\rr_1),\ldots, \s_i = f_i(\X,\s_1,\ldots,\s_{i-1}, \rr_{1},\ldots,\rr_{i-1}), \rr_i = g_i(\Y,\s_1,\ldots,\s_{i},\rr_{i,},\ldots,\rr_{i-1}),\ldots,\rr_{\ell} = g_{\ell}(\Y,\s_1,\ldots,\s_{\ell},\rr_{1},\ldots,\rr_{\ell-1})$. 

Finally, $\X' = f(\X,\rr_1,\ldots,\rr_{\ell},\s_1,\ldots,\s_{\ell})$ and $\Y' = g(\Y,\rr_1,\ldots,\rr_{\ell},\s_1,\ldots,\s_{\ell})$ correspond to the output of Alice and  the output of Bob respectively. Thus,  $h(\X,\Y) = (\X',\Y')$.

Similar to the way we argue about alternating extraction protocols, we fix random variables in the following order: Fix $\s_1$, and it follows that $\rr_1$ is now a deterministic function of $\Y$. We fix $\rr_1$, and thus $\s_2$ is now a deterministic function of $\X$. Thus, continuing in this way, we can fix all the random variables $\s_1,\ldots,\s_{\ell}$ and $\rr_1,\ldots,\rr_{\ell}$ while maintaining that $\X$ and $\Y$ are independent. Further, invoking Lemma $\ref{lemma:entropy_loss_1}$, with probability at least $1-2^{-\Omega(n)}$, both $\X$ and $\Y$ have min-entropy at least $n-t - \delta n \ge n - 2 \delta n$ since both parties send at most $t$ bits. 

Note that now, $\X'$ is a deterministic function of $X$ and $\Y'$ is a deterministic function of $Y$. Thus, any invertible $2$-source non-malleable extractor  for min-entropy $n-2\delta n$ can be used. Our result follows by using such a construction from a recent work of Li \cite{Li18}.
 
\paragraph{Extractors for interleaved sources.}
We construct an explicit extractor $\ilext:\zo^{2n} \rightarrow \zo^{m}$, $m=\Omega(n)$ that satisfies the following: Let $\X$ and $\Y$ be independent $(n,k)$-sources with $k\ge (2/3 + \delta)n$, for any constant  $\delta>0$. Let $\pi:[2n] \rightarrow [2n]$ be any permutation. Then, $$ |\ilext((\X \circ \Y)_{\pi}) - \U_m| \le \epsilon.$$ 
We present our construction and also explain the proof along the way, as this gives more intuition to the different steps of the construction.  Let $\Z = (\X \circ \Y)_{\pi}$. We start by taking a large enough slice $\Z_1$ from $\Z$ (say, of length $(2/3 + \delta/2)n$). Let $\X$ have more bits in this slice than $\Y$. Let $\X_1$ be the bits of $\X$ in $\Z_1$ and $\X_2$ be the remaining bits of $\X$. Similarly define $\Y_1$ and $\Y_2$. Notice that $\X_1$ has linear entropy and also that  $\X_2$ has linear entropy conditioned on $\X_1$.  We fix $\Y_1$ and use a condenser (from  works of Barak et al. \cite{BRSW12} and Zuckerman \cite{Zuck07}) to condense $\Z_1$ into a matrix with a constant number such that at least one of the row has entropy rate at least $0.9$. Notice that this matrix is a deterministic function of $\X$. The next step is to use $\Z$ and each row of the matrix as a seed to a linear seeded extractor get longer rows. This requires some care for the choice of the linear seeded extractor since the seed has some deficiency in entropy. After this step, we use the advice correlation breaker  from \cite{CL15} on $\Z$ and each row of the somewhere random source with the row number as the advice (similar to as done before in the construction of seedless non-malleable extractors for $2\iss$), and compute the bit-wise  XOR of the different outputs that we produce. Let $\V$ denote this random variable. Finally, to output $\Omega(n)$ bits we use a linear seeded extractor on $\Z$ with $\V$ as the seed. The correctness of various steps in the proof exploit the fact that $\Z$ can be written as the bit-wise sum of two independent sources, and the fact that we use linear seeded extractors. We refer the reader to Section $\ref{sec:ilext}$ for more details.

\section{Background and notation} \label{section:prelims}
 We use $\U_m$ to denote the uniform distribution on $\{0,1 \}^m$.  \newline For any integer $t>0$, $[t]$ denotes the set $\{1,\ldots,t \}$.\newline For a string $y$ of length $n$, and any subset $S \subseteq [n]$, we use $y_S$ to denote the projection of $y$ to the coordinates indexed by $S$. \newline We use bold capital letters for random variables and  samples as the corresponding small letter, e.g., $\X$ is a random variable, with $x$ being a sample of $\X$. \newline For strings $x, y \in \zo^{n}$, we use $x+y$ (or equivalently $x-y$) to denote the bit-wise xor of the two strings.

\subsection{A probability lemma}
The following result on min-entropy was proved by Maurer and  Wolf \cite{MW07}.
\begin{lemma}\label{lemma:entropy_loss_1} Let $\X,\Y$ be random variables such that the random variable $\Y$ takes at $\ell$ values. Then 
\begin{align*}
 \pr_{y \sim \Y}[ H_{\infty}(\X| \Y = y) \ge H_{\infty}(\X) - \log \ell -\log(1/\epsilon)] > 1-\epsilon.
 \end{align*}
\end{lemma} 
\subsection{Conditional min-entropy}
\begin{define} The average conditional min-entropy of a source $\X$ given a random variable $\W$ is defined as $$ \widetilde{H}_{\infty}(\X|\W) = -\log \l( \E_{w \sim W}\l[\max_{x} \Pr[\X=x | \W=w] \r] \r) = - \log \l(\E\l[ 2^{-H_{\infty}(\X|\W=w)} \r]\r).$$
\end{define}
We recall some results on conditional min-entropy from the work of Dodis et al.\ \cite{DORS08}.
\begin{lemma}[\cite{DORS08}] 
\label{lem:avg_worst_min}
For any $\epsilon>0$, $$\pr_{w \sim \W}\l[H_{\infty}(\X|\W=w) \ge \widetilde{H}_{\infty}(\X|\W)-\log(1/\epsilon)\r] \ge 1- \epsilon.$$
\end{lemma}
\begin{lemma}[\cite{DORS08}]\label{lem:entropy_loss} If a random variable $\Y$ has support of size $2^\ell$, then $\widetilde{H}_{\infty}(\X|\Y) \ge H_{\infty}(\X) - \ell$.
\end{lemma}

\begin{define}A function $\Ext:\{0,1\}^{n} \times \{ 0,1\}^d \rightarrow \{ 0,1\}^m$ is a $(k,\epsilon)$-seeded extractor  if for any source $\X$ of min-entropy $k$, $|\Ext(\X,\U_d) - \U_m| \le \epsilon$.  $\Ext$ is called a strong seeded extractor if $|(\Ext(\X,\U_d), \U_d) - (\U_m,\U_d) | \le \epsilon$, where $\U_m$ and $\U_d$ are independent. 

Further, if for each $s\in \U_d$, $\Ext(\cdot,s):\{ 0,1\}^n\rightarrow \{ 0,1\}^m$ is a linear function, then $\Ext$ is called a linear seeded extractor.
\end{define}

We require extractors that can extract  uniform bits when the source only has sufficient conditional min-entropy. 
\begin{define} A $(k,\epsilon)$-seeded average case seeded extractor $\Ext:\{ 0,1\}^n \times \{ 0,1\}^d \rightarrow \{ 0,1\}^m$ for min-entropy $k$ and error $\epsilon$ satisfies the following property:  For any source $\X$ and any arbitrary random variable $\Z$ with $\tilde{H}_{\infty}(\X|\Z)\ge k$, $$\Ext(\X,\U_d),\Z \approx_{\epsilon} \U_m, \Z.$$ 
\end{define}
It was shown in \cite{DORS08} that any seeded extractor is also an average case extractor.
\begin{lemma}[\cite{DORS08}]\label{lem:cond_ext} For any $\delta>0$, if $\Ext$ is a  $(k,\epsilon)$-seeded extractor, then it is also a  $(k+\log(1/\delta),\epsilon+\delta)$-seeded average case extractor.
\end{lemma}
\subsection{Samplers and extractors}\label{sec:samp_weak}
Zuckerman \cite{Z97} showed that seeded extractors can be used as samplers given access to weak sources. This connection is best presented by a graph theoretic representation of seeded extractors. A seeded extractor $\Ext:\zo^n \times \zo^d \rightarrow \zo^m$ can  be viewed as an unbalanced bipartite graph $G_{\Ext}$ with $2^n$ left vertices (each of degree $2^d$) and $2^m$ right vertices. Let $\N(x)$ denote the set of neighbors of $x$ in $G_{\Ext}$.

\begin{thm}[\cite{Z97}]\label{bad_set}Let $\Ext:\{0,1\}^{n} \times \{ 0,1\}^{d} \rightarrow \{ 0,1\}^{m}$ be a  seeded extractor for min-entropy $k$ and error $\epsilon$. Let $D=2^d$. Then for any set $R \subseteq \{0,1\}^{m}$, $$ |\{x \in \{ 0,1\}^n : | |\N(x) \cap R| - \mu_R D| > \epsilon D \}| < 2^k,$$  where $\mu_R = |R|/2^{m}$.
\end{thm}

\begin{thm}[\cite{Z97}]\label{thm:seed_samp}Let $\Ext:\{0,1\}^{n} \times \{ 0,1\}^{d} \rightarrow \{ 0,1\}^{m}$ be a  seeded extractor for min-entropy $k$ and error $\epsilon$.  Let $\{ 0,1\}^d=\{r_1,\ldots,r_D\}$, $D=2^d$. Define  $\samp(x) = \{\Ext(x,r_1),\ldots,\Ext(x,r_D)\}$. Let $\X$ be an $(n,2k)$-source. Then for any set $R \subseteq \{0,1\}^{m}$, $$\pr_{\xb \sim \X}[||\samp(\xb) \cap R | - \mu_RD| > \epsilon D] < 2^{-k},$$ where $\mu_R = |R|/2^{m}$.
\end{thm}

\subsection{Explicit extractors from prior work}
We recall an optimal construction of strong-seeded extractors.
\begin{thm}[\cite{GUV09}]\label{guv} For any constant $\alpha>0$, and all integers $n,k>0$ there exists a polynomial time computable  strong-seeded extractor $\Ext: \{ 0,1\}^n \times \{ 0,1\}^d \rightarrow \{ 0,1\}^m$   with $d = O(\log n + \log (1/\epsilon))$ and $m = (1-\alpha)k$.
\end{thm}  

The following  are explicit constructions of linear seeded extractors.
\begin{thm}[\cite{Tr01,RRV02}]\label{trev_ext} For every $n,k,m \in \mathbb{N}$ and $\epsilon>0$, with $m \le k \le n$, there exists an explicit strong linear seeded extractor $\LExt:\{ 0,1\}^n \times \{ 0,1\}^d \rightarrow \{ 0,1\}^m$ for min-entropy $k$ and error~$\epsilon$, where $d = O\l(\log^2(n/\epsilon)/\log(k/m)\r)$.
\end{thm}

A drawback of the above construction is that the seeded length is $\omega(\log n)$ for sub-linear min-entropy. A construction of Li \cite{Li:affine} achieves $O(\log n)$ seed length for even polylogarithmic min-entropy.
\begin{thm}[\cite{Li:affine}]\label{li_ext} There exists a constant $c>1$ such that for every $n,k \in \mathbb{N}$  with $c\log^8 n \le k \le n$ and any $\epsilon \ge 1/n^2$, there exists a polynomial time computable linear seeded extractor $\LExt: \{ 0,1\}^n \times \{ 0,1\}^d \rightarrow \zo^m$ for min-entropy $k$ and error~$\epsilon$, where $d= O(\log n)$ and $m \le  \sqrt{k}$.
\end{thm}

A different construction achieves seed length $O(\log(n/\epsilon))$ for high entropy sources.
\begin{thm}[\cite{CGL15,Li16}] 
\label{thm:low_error_inv_lin}
For all $\delta>0$ there exist $\alpha,  \gamma>0$ such that for all integers $n>0$, $\epsilon \ge 2^{-\gamma n}$, there exists an efficiently computable linear strong seeded extractor $\LExt:\zo^n \times \zo^d \rightarrow \zo^{\alpha d}$, $d = O(\log (n/\epsilon))$ for min-entropy $\delta n$. Further, for any $ y \in \zo^{d}$, the linear map $\LExt(\cdot,y)$ has rank $\alpha d$.
\end{thm}
The above theorem is stated in \cite{Li16} for $\delta = 0.9$, but it is straightforward to see that the proof extends for any constant $\delta>0$.

We use a property of linear seeded extractors proved by Rao \cite{Rao09}.
\begin{lemma}[\cite{Rao09}]\label{aff_error} Let $\Ext:\{ 0,1\}^{n} \times \{ 0,1\}^{d} \rightarrow \{ 0,1\}^{m}$ be a linear seeded extractor for min-entropy $k$ with error $\epsilon<\frac{1}{2}$. Let $X$ be an affine $(n,k)$-source. Then $$\Pr_{u \sim U_{d}}[|\Ext(X,u) - U_{m}|>0] \le 2\epsilon. $$
\end{lemma}

We recall a two-source extractor construction for high entropy sources based on the inner product function.
\begin{thm}[\cite{CG88} ]
\label{thm:strong_ip} For all $m,r>0$, with $q=2^{m}, n =rm$, let $\X,\Y$ be independent sources on $\bb{F}_q^r$ with min-entropy $k_1,k_2$ respectively. Let $\IP$ be the inner product function over the field $\bb{F}_q$.  Then, we have: $$|\IP(\X,\Y), \X - \U_m, \X| \le \epsilon, \hspace{0.5cm} |\IP(\X,\Y), \Y - \U_m, \Y| \le \epsilon$$ where $\epsilon =2^{-(k_1+k_2-n-m)/2}$.
\end{thm}

\subsection{Advice correlation breakers}
\label{sec:acb}
We use a primitive called `correlation breaker' in our construction. Consider a situation where we have arbitrarily correlated random variables $\Y^1,\ldots,\Y^r$, where each $\Y^i$ is on $\ell$ bits. Further suppose $\Y^1$ is a `good' random variable (typically, we assume $\Y^1$ is uniform or  has almost full min-entropy). A correlation breaker $\cb$ is an explicit function that takes some additional resource $\X$, where $\X$ is typically additional randomness (an $(n,k)$-source) that is independent  of $\{\Y^1,\ldots,\Y^r\}$. Thus using $\X$, the task is to  break the correlation between $\Y^1$ and the random variables $\Y^2,\ldots,\Y^r$, i.e., $\cb(\Y^1,\X)$ is independent of $\{\cb(\Y^2,\X),\ldots,\cb(\Y^r,\X)\}$. A weaker notion is that of an advice correlation breaker  that takes in some advice for each of the $\Y^i$'s as an additional resource in breaking the correlations. This primitive was implicitly constructed in \cite{CGL15} and used in explicit constructions of non-malleable extractors, and has subsequently found many applications in explicit constructions of extractors for independent sources and non-malleable extractors.

We recall an explicit advice correlation breaker constructed in \cite{CL15}. This correlation breaker works even with the weaker guarantee that the `helper source' $\X$ is now allowed to be correlated to the sources random variables $\Y^1,\ldots,\Y^r$ in a structured way. Concretely, we assume the source to be of the form  $\X + \Z$, where $\X$ is assumed to be an $(n,k)$-source that is uncorrelated with $\Y^1,\ldots,\Y^r, \Z$. We now state the result more precisely. 
\begin{thm}[\cite{CL15}]\label{thm:acb} For all integers $n,n_1,n_2, k, k_1,k_2,t,d,h,\la$  and any $\epsilon>0$,  such that  $d=O(\log^2(n/\epsilon))$,  $k_1 \ge 2d+  8tdh + \log(1/\epsilon)$,  $n_1\ge2d + 10tdh + (4h t +1)n_2^2+\log(1/\epsilon)$, and $n_2 \ge 2d +3td+\log(1/\epsilon)$, let
\begin{itemize}
\item $\X$ be an $(n,k_1)$-source, $\X'$ a  r.v on n bits,  $\Y^1$ be an $(n_1,n_1-\la)$-source, $\Z,\Z'$ are r.v's on $n$ bits, and $\Y^{2},\ldots,\Y^{t}$ be r.v's on $n_1$ bits each, such that $\{\X,\X'\}$ is independent of $\{\Z,\Z',\Y^{1},\ldots,\Y^{t}\}$,
\item $id^1,\ldots,id^{t}$ be  bit-strings of length $h$ such that for each $i\in \{2,t\}$, $id^1 \neq id^{i}$. 
\end{itemize}
Then there exists an efficient algorithm $\acb:\zo^{n_1} \times \zo^{n} \times \zo^{h} \rightarrow \zo^{n_2}$ which satisfies the following: let 
\begin{itemize}
\item $\Y^{1}_{h}=\acb(\Y^1,\X+\Z,id^1)$,
\item $\Y^{i}_{h}=\acb(\Y^i,\X'+\Z',id^i)$, $i\in [2,t]$
\end{itemize}
Then, $$\Y^{1}_{h},\Y^{2}_{h},\ldots,\Y^{t}_{h}, \X, \X' \approx_{O((h+2^{\la})  \epsilon)} \U_{n_2},  \Y^{2}_{h},\ldots,\Y^{t}_{h}, \X, \X'.$$
\end{thm}

\section{NM extractors for linear composed with interleaved split-state adversaries}
\label{section:composed_split_state}
The main result of this section is an explicit non-malleable extractor against the tampering family $\lin \circ 2\iss \subset \F_{2n}$. 
\begin{thm}
\label{theorem:ext_lin_composed_ss_1}
For all integers $n>0$  there exists an explicit function $\nmExt: \zo^{2n} \rightarrow \zo^{m}$,   $m=n^{\Omega(1)}$, such that the following holds: For any linear function $h : \zo^{2n} \rightarrow \zo^{2n}$, arbitrary tampering functions $f,g \in \F_n$,  any permutation $\pi:[2n] \rightarrow [2n]$ and  independent  uniform sources $\X$ and $\Y$ each on $n$ bits, there exists a distribution $\D_{h,f,g,\pi}$ on $\zo^m \cup  \{ \same\}$, such that 
$$ |\nmExt((\X \circ \Y)_{\pi}),  \nmExt(h((f(\X) \circ g(\Y))_{\pi})) - \U_{m},\cpy(\D_{h,f,g,\pi},\U_m) | \le 2^{-n^{\Omega(1)}}.$$
\end{thm}
Our first step is to show that in order to prove Theorem $\ref{theorem:ext_lin_composed_ss_1}$  it is enough to construct a non-malleable extractor satisfying Theorem $\ref{theorem:ext_lin_composed_ss_2}$.

\begin{thm}
\label{theorem:ext_lin_composed_ss_2}
There exists a $\delta>0$ such that for all integers $n,k>0$ with $n \ge k \ge n - n^{\delta}$, there exists an explicit function $\nmExt: \zo^{2n} \rightarrow \zo^{m}$,   $m=n^{\Omega(1)}$, such that the following holds: Let $\X$ and $\Y$  be independent $(n,n-n^{\delta})$-sources, $\pi:[2n] \rightarrow [2n]$ any arbitrary permutation and arbitrary tampering functions $f_1,f_2,g_1,g_2 \in \F_n$ that satisfy the following condition: 
\begin{itemize}
\item $\forall x \in support(\X)$ and $y \in support(\Y)$, $f_1(x) + g_1(y) \neq x$ or 
\item $\forall x \in support(\X)$ and $ y \in support(\Y)$, $f_2(x) + g_2(y) \neq y$.
\end{itemize} Then, 
\begin{align*}
|\nmExt((\X \circ \Y)_{\pi}),  \nmExt(((f_1(\X) + g_1(\Y)) \circ (f_2(\X) + g_2(\Y)))_{\pi}) -  \\ \U_{m},  \nmExt(((f_1(\X) + g_1(\Y)) \circ (f_2(\X) + g_2(\Y)))_{\pi})| &\le 2^{-n^{\Omega(1)}}.
\end{align*}
\end{thm}
\begin{proof}[Proof of Theorem $\ref{theorem:ext_lin_composed_ss_1}$ assuming Theorem $\ref{theorem:ext_lin_composed_ss_2}$] Define  $\ol{f(x)}= h((f(x) \circ 0^n)_{\pi})$ and $\ol{g(y)} = h((0^n \circ y)_{\pi})$. Thus, $h((f(x) \circ g(y))_{\pi}) = \ol{f(x)} + \ol{g(y)}$. Define functions $h_1:\zo^{2n} \rightarrow \zo^n$ and    $h_2:\zo^{2n} \rightarrow \zo^n$ such that $h((f(x) \circ g(y))_{\pi})= (h_1(x,y) \circ h_2(x,y))_{\pi}$. Since $h(f(x), g(y)) = \ol{f(x)} + \ol{g(y)}$, it follows that there exists functions $f_1,g_1, f_2, g_2 \in \F_n$ such that for all $x,y \in \zo^n$, the following hold: 
\begin{itemize}
\item $h_1(x,y) = f_1(x) + g_1(y)$, and
\item $h_2(x,y) = f_2(x) + g_2(y)$.
\end{itemize}
Thus, $h((f(x) \circ g(y))_{\pi}) = ((f_1(x) + g_1(y)) \circ (f_2(x)+ g_2(y)))_{\pi}$. 

Now, the idea is to show that $((\X \circ \Y)_{\pi}, ((f_1(\X)+g_1(\Y)) \circ(f_2(\X) + g_2(\Y)))_{\pi})$ is $2^{-n^{\Omega(1)}}$-close to a convex combination of $((\X \circ \Y)_{\pi}, (\X \circ \Y)_{\pi})$ and distributions of the form $((\X' \circ \Y')_{\pi}, ((\eta_1(\X)+\nu_1(\Y)) \circ (\eta_2(\X) + \nu_2(\Y)))_{\pi})$, where $\X'$ and $\Y'$ are independent $(n,n-n^{\delta})$-sources and $\eta_1,\eta_2,\nu_1,\nu_2$ are deterministic functions in $\F_n$ satisfying the conditions that: 
\begin{itemize}
\item $\forall x \in support(\X')$ and $y \in support(\Y')$, $\eta_1(x) + \nu_1(y) \neq x$ or 
\item $\forall x \in support(\X')$ and $ y \in support(\Y')$, $\eta_2(x) + \nu_2(y) \neq y$.
\end{itemize}
Theorem $\ref{theorem:ext_lin_composed_ss_1}$ is then direct from from Theorem $\ref{theorem:ext_lin_composed_ss_2}$.

Let $n_0 = n^{\delta}$. For any $y \in \zo^{n}$ and any function $\eta: \zo^{n} \rightarrow \zo^{n}$, let $\eta^{-1}(y)$ denote the set $\{ z \in \zo^{n}: \eta(z) = y\}$. We partition $\zo^n$ into the following two sets: $$\Gamma_1 = \{ y \in \zo^{n}: |g_1^{-1}(g_1(y))| \ge 2^{n-n_0}\}, \hspace{1cm}\Gamma_2 = \zo^{n} \setminus \Gamma_1.$$ 
Let $\Y_1$ be uniform on $\Gamma_1$ and $\Y_2$ be uniform on $\Gamma_2$. Clearly, $\Y$ is a convex combination of $\Y_1$ and $\Y_2$ with weights $w_i = |\Gamma_1|/2^n$, $i =1,2$. If $w_{i} \le 2^{-n_0/2}$, we ignore the corresponding source and add an error of $2^{-n_0/2}$ to the extractor. Thus, suppose $w_i \ge 2^{-n_0/2}$  for $i=1,2$. Thus, $\Y_1$ and $\Y_2$ each have min-entropy at least $n- n_0/2$.

We claim that $g_1(\Y_2)$ has min-entropy at least $n_0/2$. This can be seen in the following way. For any $y \in \Gamma_2$, $|g_1^{-1}(g_1(y))| \le 2^{n-n_0}$, and hence it follows $g_1(\Y_2)$ has min-entropy at least $(n-n_0/2)-(n-n_0) = n_0/2$. Thus, clearly for any $x \in \zo^n$, $x + g_1(\Y_2) \neq x$ with probability at least $1-2^{-n_0/2}$. We add a term of $2^{-n^{\Omega(1)}}$ to the error and assume that $\X+g_1(\Y_2) \neq \X$. Thus, $(\X \circ \Y_2)_{\pi},((f_1(\X)+g_1(\Y_2)) \circ (f_1(\X)+g_1(\Y_2)))_{\pi}$ is indeed $2^{-n^{\Omega(1)}}$close to a convex combination of distributions of the required form.

Next, we claim that for any fixing of $g_1(\Y_1)$, the random variable $\Y_1$ has min-entropy at least $n-n_0$. This is direct from the fact that for any $y \in \Gamma_2$, $|g_1^{-1}(g_1(y))| > 2^{n-n_0}$. We fix $g_1(\Y_1)=g$, and let $f_{1,g}(x) = f_1(x) + g$. Thus, $f_{1,g}(\X) = f_1(\X) + g_1(\Y_1)$. We now partition $\zo^n$ according to the fixed points of $f_{1,g}$. Let $$\Delta_1 = \{x: f_1'(x) = x \}, \hspace{1cm}\Delta_2 = \zo^n \setminus \Delta_1.$$ 

Let $\X_{1}$ be a flat distribution on $\Delta_1$ and $\X_2$ be a flat distribution on $\Delta_2$. If $|\Delta_1| < 2^{n-n_0/2}$, we ignore the distribution $\X_1$ and add an error of $2^{n-n_0/2}$ to the analyis of the non-malleable extractor. Further, it is direct from definition that $f_1(\X_2) + g \neq \X_2$. We now handle to case when $\Delta_1 > 2^{n-n_0/2}$. Note that in this case, $H_1(\X_1) \ge n- n_0/2$. The idea is now to partition $\Delta_1$ into two sets based on the pre-image size of $f_2$ similar to the way we partioned the support of $\Y$ based on the pre-image size of $g_1$. Define the sets $$\Delta_{11} = \{ x \in \Delta_1: |f_2^{-1}(f_2(x)) \cap \Delta_1 | \ge 2^{n-n_0}\}, \hspace{1cm}\Delta_{12} = \Delta_1 \setminus \Delta_{11}.$$ 

Let $\X_{11}$ be flat on $\Delta_{11}$ and $\X_{12}$ be flat on $\Delta_{12}$. Clearly, $\X_1$ is a convex combination of the sources $\X_{11}$ and $\X_{12}$. If $\Delta_{11}$ or $\Delta_{12}$ is smaller than $2^{n-3n_0/4}$, we ignore the corresponding distribution and add an error of $2^{-n_0/4}$ to the error analysis of the non-malleable extractor. Thus suppose $\Delta_{1i}\ge 2^{n-3n_0/4}$ for $i=1,2$. Thus, $\X_{11}$ and $\X_{12}$ both have min-entropy at least $n-3n_0/4$.

We claim that $f_2(\X_{12})$ has min-entropy at least $n_0/4$. This can be seen in the following way. For any $x \in \Delta_{12}$, $|f_2^{-1}(f_2(x)) \cap \Delta_1| \le 2^{n-n_0}$, and hence it follows $f_2(\X_{12})$ has min-entropy at least $(n-3n_0/4)-(n-n_0) = n_0/4$. Thus, clearly $f_2(\X_{12}) + g_2(\Y_1) \neq \Y_1$ with probability at least $1-2^{-n_0/4}$. As before, we add an error of $2^{-n^{\Omega(1)}}$ to the error, and assume that $f_2(\X_{12}) + g_2(\Y_1) \neq \Y_1$. Thus, $(\X_{12} \circ \Y_1)_{\pi},((f_1(\X_{12})+g_1(\Y_2)) \circ (f_1(\X_{12})+g_1(\Y_2)))_{\pi}$ is indeed $2^{-n^{\Omega(1)}}$-close to a convex combination of distributions of the required form.

Next, we claim that for any fixing of $f_2(\X_{11})$, the random variable $\X_{11}$ has min-entropy at least $n-n_0$. This is direct from the fact that for any $x \in \Delta_1$, $|f_2^{-1}(f_1(x)) \cap \Delta_1| > 2^{n-n_0}$. We fix $f_2(\X_{11})=\la$, and let $g_{2,\la}(y) = \la + g_2(y)$. Thus, $g_{2,\la}(\Y) = f_1(\X) + g_1(\Y_1)$. We now partition $\Gamma_1$ according to the fixed points of $f_{1,g}$. Let $$\Gamma_{11} = \{y: g_{2,\la}(y) = y \}, \hspace{1cm}\Gamma_{12} = \zo^n \setminus \Gamma_{11}.$$ 

Let $\Y_{11}$ be a flat distribution on $\Gamma_{11}$ and $\Y_{12}$ be a flat distribution on $\Gamma_{12}$. It follows from definition that $(f_1(\X_{11}) + g_1(\Y_{11}), f_2(\X_{11}) + g_2(\Y_{11})) = (\X_{11}, \Y_{11})$. Further, 
$f_2(\X_{11}) + g_2(\Y_{12}) \neq \Y_{12}$, and hence $(\X_{11} \circ \Y_{12})_{\pi}, ((f_1(\X_{11})+g_1(\Y_{12})) \circ (f_1(\X_{11})+g_1(\Y_{12})))_{\pi}$ is $2^{-n^{\Omega(1)}}$-close to a convex combination of distributions of the required form. This completes the proof. 
\end{proof}
In the rest of the section, we prove  Theorem $\ref{theorem:ext_lin_composed_ss_2}$. We assume the setup given from Theorem $\ref{theorem:ext_lin_composed_ss_2}$. Thus, $\X$ and $\Y$  are independent $(n,n-n^{\delta})$-sources, $\pi:[2n] \rightarrow [2n]$ is an arbitrary permutation and $f_1,f_2,g_1,g_2 \in \F_n$  satisfy the following conditions: 
\begin{itemize}
\item $\forall x \in support(\X)$ and $y \in support(\Y)$, $f_1(x) + g_1(y) \neq x$ or 
\item $\forall x \in support(\X)$ and $ y \in support(\Y)$, $f_2(x) + g_2(y) \neq y$.
\end{itemize}

We use the following notation: if $\W = h((\X \circ \Y)_{\pi})$ (for some function $h$), then we use to $\W'$ or $(\W)'$ to denote the random variable $h(((f_1(\X) + g_1(\Y)) \circ (f_2(\X) + g_2(\Y)))_{\pi})$.  Further,  define  $\overline{\X}= (\X \circ 0^n)_{\pi}$, $\overline{\Y} = (0^n \circ \Y )_{\pi}$, $\overline{f_1(\X)}=(f_1(\X) \circ 0^n)_{\pi}$, $\overline{f_2(\X)}=(0^n \circ f_2(\X))_{\pi}$,  $\overline{g_1(\Y)}= (g_1(\Y) \circ 0^n)_{\pi}$ and $\overline{g_2(\Y)}= (0^n \circ g_2(\Y))_{\pi}$. It follows that $\Z= \overline{\X} + \overline{\Y}$ and $\Z' = \overline{f_1(\X)} + \overline{g_1(\Y)} + \overline{f_2(\X)} + \overline{g_2(\Y)}$.

We use Section $\ref{sec:new_adv_1}$ to construct an advice generator and Section~$\ref{sec:new_ext_composed}$ to construct an advice correlation breaker. Finally, we present the non-malleable extractor construction in Section~$\ref{sec:new_ext_composede}$.
\subsection{An advice generator}
\label{sec:new_adv_1}
\begin{lemma}\label{lem:new_adv_1}There exists   an efficiently computable function $\adv:\zo^{n} \times \zo^{n} \rightarrow \zo^{n_4}$, $n_4=  n^{\delta}$, such that  with probability at least $1-2^{-n^{\Omega(1)}}$ over the fixing of the random variables $\{\adv((\X \circ \Y)_{\pi}),  \adv(((f_1(\X) + g_1(\Y)) \circ (f_2(\X)+g_2(\Y)))_{\pi})\}$, the following hold:

\begin{itemize}
\item $\{\adv((\X \circ \Y)_{\pi}) \neq \adv(((f_1(\X) + g_1(\Y)) \circ (f_2(\X)+g_2(\Y)))_{\pi})\}$,
\item $\X$ and $\Y$ are independent,
\item  $H_{\infty}(\X) \ge k -2 n^{\delta} $, $H_{\infty}(\Y) \ge k - 2n^{\delta}$.
\end{itemize}
\end{lemma}
We prove the above lemma in the rest of this subsection.  We claim that the function $\adv$ computed by Algorithm $\ref{alg:advice_new}$ satisfies the above lemma. We first set up some parameters and ingredients. 
\begin{itemize}
\item Let $C$ be a large enough constant and $\delta'=\delta/C$.
\item Let $n_0 = n^{\delta'}, n_1=  n_0^{c_0},  n_2= 10n_0$, for some  constant $c_0$ that we set below.
\item Let $E:\zo^{2n} \rightarrow \zo^{n_3}$ be the encoding function of a linear error correcting code  $\C$ with constant rate $\alpha$ and constant distance $\beta$.
\item Let $\Ext_1: \zo^{n_1} \times \zo^{d_1} \rightarrow \zo^{\log(n_3)}$ be a $(n_1/20,\beta/10)$-seeded extractor instantiated using Theorem $\ref{guv}$. Thus $d_1= c_1 \log n_1$, for some constant $c_1$. Let $D_1=2^{d_1}=n_1^{c_1}$.
\item Let $\samp_1:\zo^{n_1} \rightarrow [n_3]^{D_1}$ be the sampler obtained from Theorem $\ref{thm:seed_samp}$ using $\Ext_1$. 
\item Let $\Ext_2: \zo^{n_2} \times \zo^{d_2} \rightarrow \zo^{\log(2n)}$ be a $(n_2/20,1/n_0)$-seeded extractor instantiated using Theorem $\ref{guv}$. Thus $d_2= c_2 \log n_2$, for some constant $c_2$. Let $D_2=2^d_2$.  Thus $D_2= 2^{d_2}= n_2^{c_2}$.
\item Let $\samp_2:\zo^{n_2} \rightarrow [2n]^{D_2}$ be the sampler obtained from Theorem $\ref{thm:seed_samp}$ using $\Ext_2$.
\item Set $c_0 = 2c_2$.
\item Let $\ilext:\zo^{D_2} \rightarrow \zo^{n_0}$ be the extractor from Theorem $\ref{thm:il_ext}$.
\item Let $\LExt: \zo^{2n} \times \zo^{n_0} \rightarrow \zo^{n_0}$ be a linear seeded extractor instantiated from Theorem $\ref{thm:strong_ip}$ set to extract from min-entropy $n_1/100$ and error $2^{-\Omega(\sqrt{n_0})}$ .
\end{itemize}

\RestyleAlgo{boxruled}
\LinesNumbered
\begin{algorithm}[ht]\label{alg:advice_new}
  \caption{$\adv(z)$  \vspace{0.1cm}\newline \textbf{Input:} 
  Bit-string  $z= (x \circ y)_{\pi}$ of length $2n$, where $x$ and $y$ are each $n$ bit-strings and $\pi: [2n] \rightarrow [2n]$ is a permutation. \newline \textbf{Output:} Bit string $v$  of length $n_4$. 
 }
Let $z_1 = \slice(z,n_1), z_2= \slice(z,n_2)$. 

Let $S= \samp_1(z_1)$.

Let $T = \samp_2(z_2)$ and $z_3 = z_{T}$.

Let $r = \ilext(z_3)$.

Let $w_{1} = (E(z))_S$.

Let  $w_{2} = \LExt(z,r)$.

Output $v = z_1 \circ z_2 \circ z_3 \circ w_1 \circ w_2$.
\end{algorithm}

\begin{lemma}\label{lem:advice_new_2} With probability at least $1-2^{-n^{\Omega(1)}}$, $ \V \neq \V'$.
\end{lemma}
\begin{proof} We prove the lemma assuming $f_1(\X) + g_1(\Y) \neq \X$. The proof in the other case (i.e., $f_2(\X) + g_2(\Y) \neq \Y$) is similar and we skip it. 

First observe that the lemma is direct if  $\Z_1 \neq \Z_1'$  or  $\Z_2 \neq \Z_2'$ or $\Z_3 \neq \Z_3'$. Thus, we can assume $\Z_i = \Z_i'$ for $i=1,2,3$.  It is easy to see that $\s = \s', \T= \T'$, and $\Z_4 = \Z_4'$. 

 Now observe that 
 \begin{align*}
 \Z - \Z' = \overline{\X} + \overline{\Y} -  \overline{f_1(\X)} - \overline{g_1(\Y)} - \overline{f_2(\X)} - \overline{g_2(\Y)}.
 \end{align*}
 Note that $\Z - \Z' \neq 0$ which follows from our assumption that $f_1(\X) + g_1(\Y) \neq \X$.
 
 Now define the function $h_1: \zo^{2n} \rightarrow \zo^{2n}$ as $h_1(z) = z - f_1(z) - f_2(z)$ and $h_2: \zo^{2n} \rightarrow \zo^{2n}$ as $h_2(z) = z - g_1(z) - g_2(z)$.
 
 Thus, 
\begin{align*}
 \Z - \Z' = h_1(\overline{\X}) + h_2(\overline{\Y}).
  \end{align*}
 Let $\X_i$ be the bits of $\X$ in $\Z_i$ for $i=1,2,3$ and $\X_4$ be the remaining bit of $\X$. Similarly define $\Y_i$'s, $i=1,2,3,4$. 
Without loss of generality suppose that $|\X_1| \ge |\Y_1|$, (where $|\alpha|$ denotes the length of the string $\alpha$).

Let $\Gamma \subset \zo^{2n}$ denote the support of the source $\overline{\X}$. We partition $\Gamma$ into two sets $\Gamma_a$ and $\Gamma_b$ according to the pre-image  size of the function $h_1$ in the following way. For any $z \in \zo^{2n}$, let $h_1^{-1}(z)$ denote the set $\{ y \in \zo^{2n}: h_1(y) = z\}$. 

Let $n_{p} = n_1/50$.
 Define $$\Gamma_a = \{ z \in \Gamma: |h_1^{-1}(h_1(z)) \cap \Gamma| \ge 2^{n-n_p} \}, \hspace{0.5cm}\Gamma_b = \Gamma \setminus \Gamma_1.$$ 
 Let $p_a = \Pr[\overline{\X} \in \Gamma_a]$ and $p_b = \Pr[\overline{\X} \in \Gamma_b]$.
  Let $\overline{\X}_a$ be the source supported on $\Gamma_a$ with the probability law $\Pr[\overline{\X}_a=z] = \frac{1}{p_a} \cdot \Pr[\overline{\X}=z]$. Also define $\overline{\X}_b$  supported on $\Gamma_b$ with the probability law $\Pr[\overline{\X}_a=z] = \frac{1}{p_b} \cdot \Pr[\overline{\X}=z]$. 
 
 Clearly $\overline{\X}$ is a convex combination of the distributions $\overline{\X}_a$ and $\overline{\X}_b$, with weights $p_a$ and $p_b$ respectively. If any of $p_a$ or $p_b$ is less that $2^{-n_0}$, we ignore the corresponding source and add it to the error. Thus suppose both $p_a$ and $p_b$ are at least $2^{-n_0}$. This implies that both $\overline{\X}_a$ and $\overline{\X}_b$ have min-entropy at least $n-2n_0$. We  record the following two bounds that are direct from the above definitions.
 \begin{itemize}
 \item For any fixing of $h_1(\overline{\X}_{a})=x_a$, $\overline{\X}_a$ has min-entropy at least $n-n_p$.
 \item The distribution $h_1(\X_b)$ has min-entropy at least $n_p-2n_0$.
 \end{itemize}

We introduce some notation. For any random variable $\nu = \eta(\overline{\X},\overline{\Y})$ (where $\eta$ is an arbitrary deterministic function),  we add an extra $a$ or $b$ to the subscript and  use $\nu_a$ to denote the random variable  $\eta(\overline{\X}_a,\overline{\Y})$  and $\nu_b$  to denote the random variables $\eta(\overline{\X}_b,\overline{\Y})$ respectively. For example, $\Z_{1,a}' = \overline{f_1(\X_a)} + \overline{g_1(\Y)} + \overline{f_2(\X_a)} + \overline{g_2(\Y)}$. Further we use $\X_a$ to denote the distribution on $n$ bits such that $\overline{\X}_a = (\X_a \circ 0^n)_{\pi}$. We similarly define the distribution $\X_b$.

We prove the following two statements:
\begin{enumerate}
\item $\W_{1,a} - \W_{1,a}' \neq 0 $ with probability $1-2^{-n^{\Omega(1)}}$.
\item $\W_{2,b} - \W_{2,b}' \neq 0 $ with probability $1-2^{-n^{\Omega(1)}}$.
\end{enumerate}
It is direct that the lemma follows from the above two inequalities.

We begin with the proof of $(1)$. 
Since $E$ is a linear code,  we have
\begin{align*}
\W_{1,a}  - \W_{1,a}' &= (E(\Z_a- \Z_a'))_{\s_a}.\\
			&=(E(h_1(\overline{\X_a}) + h_2(\overline{\Y})))_{\s_a}.
\end{align*}
Now fix the random $h_1(\X_a)$, and it follows that $\X_a$ has min-entropy at least $n-n_p$. Recall that we assumed $|\X_1| \ge |\Y_1|$. Thus, $\X_{1,a}$ has min-entropy at least $n_1/2 - n_p-n_0>n_1/10$ with probability at least $1-2^{-n_0}$. Further fix $\Y$, and note that this does not affect the distribution of $\X_{1,a}$. This fixes $E(\Z_a- \Z_a')$. Further $\Z_a \neq \Z_a'$, the $E(\Z_a- \Z_a')$ contains $1$'s at least $\beta$ fraction of its coordinates. Recalling that $\s_a = \samp_1(\Z_{1,a})$, it now follows from Theorem $\ref{thm:seed_samp}$ that with probability at least $1-2^{-n^{\Omega(1)}}$, $(E(\Z_a- \Z_a'))_{\s}$ is a non-zero string (and hence $\W_{1,a} - \W_{1,a}' \neq 0 $). This completes the proof of this case.

We now proceed to prove $(2)$. Using the fact that $\LExt$ is a linear seeded extractor, it follows that
\begin{align*}
\W_{2,b} - \W_{2,b}' &= \LExt(\Z_b-\Z_b',\rr_b) \\
				& =  \LExt(h_1(\X_b),\rr_b) + \LExt(h_2(\Y_b),\rr_b).
\end{align*}
Without loss of generality, suppose $\X$ has more bits in $\Z_2$ (the argument is identical in the other case).  Since $\X_{2,b}$ has min-entropy at least $n-2n_0$, it follows that $\X_{2,b}$ has min-entropy at least $\frac{n_2}{2} - 3n_0>\frac{n_2}{10}$ with probability at least $1-2^{-n_0}$. Fix the bits of $\Y$ in $\Z_2$, and thus $\Z_{2,b}$ is a deterministic function of $\X_{2,b}$. Recall that $\T_b = \samp_2(\Z_2)$. It is now straightforward to see that with probability $1-2^{-n^{\Omega(1)}}$ over the fixing of $\X_{2,b}$, $|\T_b|\cdot (1/2-o(1))\le |\T_b \cap \pi([n])| \le |\T_b| \cdot (1/2+o(1)$. Recall $|T_b| = D_2$. We fix $\X_{2,b}$ such that  $(1/2-o(1))D_2 \le |\T_b \cap \pi([n])| \le (1/2+o(1))D_2$. Thus, $\Z_3$ contains at least $(1/2-o(1))D_2$ bits from both $\X_b$ and $\Y$. It follows that both $\X_{3,b}$ and $\Y$ both have min-entropy at least  $(1/2 - o(1))D_2 - 2n_0 - n_2 = (1/2 - o(1))D_2$ (even with the conditionings so far), and hence $\rr_b$ is $2^{-n^{\Omega(1)}}$-close to uniform. We argue this this hold even conditioned on $\X_{3,b}$. This follows roughly from  the fact that any $2$-source extractor is strong \cite{rao2007exposition} which easily extends to interleaved extractors. We fix $\X_{3,b}$, and thus $\rr_b$ is now a deterministic function of $\Y$. 

Next, we note that $h_1(\X_b)$ has min-entropy at least $(n-2n_0) - (n-n_p) -n_2 - D_2-n_0=n_p - 3n_0 - D_2 - n_2 > n_p/2$ (with probability $1-2^{-n^{\Omega(1)}}$). Thus, $\LExt(h_1(\X_b),\rr_b)$ is $2^{-n^{\Omega(1)}}$-close to uniform. We fix $\rr_b$ and $\LExt(h_1(\X_b),\rr_b)$  continues to be close to uniform using the fact that $\LExt$ is a strong-seeded extractor. Further, $\LExt(h_1(\X_b),\rr_b)$ is now a deterministic function of $\X_b$ and we can fix $\LExt(h_2(\Y_b),\rr_b)$ which is a deterministic function of $\Y$. It thus follows that $\W_{2,b} - \W_{2,b}' \neq 0$ with probability $1-2^{-n^{\Omega(1)}}$ using the fact that $\LExt(h_1(\X_b),\rr_b)$ is close to uniform. This completes the proof of $(2)$. The fact that $\V$ and $\V'$ can be fixed such that $\X$ and $\Y$ remain independent with min-entropy at least $k-2n^{\delta}$ (with probability $1-2^{-n^{\Omega(1)}})$ is easy to verify from the construction. This completes the proof of  Lemma $\ref{lem:advice_new_2}$.
\end{proof}

\subsection{An Advice Correlation Breaker}
\label{sec:new_ext_composed}
We recall the setup of Theorem $\ref{theorem:ext_lin_composed_ss_2}$.   $\X$ and $\Y$  are independent $(n,k)$-sources, $k \ge n-n^{\delta}$,  $\pi:[2n] \rightarrow [2n]$ is an arbitrary permutation and $f_1,f_2,g_1,g_2 \in \F_n$  satisfy the following conditions: 
\begin{itemize}
\item $\forall x \in support(\X)$ and $y \in support(\Y)$, $f_1(x) + g_1(y) \neq x$ or 
\item $\forall x \in support(\X)$ and $ y \in support(\Y)$, $f_2(x) + g_2(y) \neq y$.
\end{itemize}
Further, we defined the following: $\overline{\X}= (\X \circ 0^n)_{\pi}$, $\overline{\Y} = ( 0^n\circ \Y)_{\pi}$, $\overline{f_1(\X)}=(f_1(\X) \circ 0^n)_{\pi}$, $\overline{f_2(\X)}=(0^n \circ f_2(\X))_{\pi}$,  $\overline{g_1(\Y)}= (g_1(\Y) \circ 0^n)_{\pi}$ and $\overline{g_2(\Y)}= (0^n \circ g_2(\Y))_{\pi}$. It follows that $\Z= \overline{\X} + \overline{\Y}$ and $\Z' = \overline{f_1(\X)} + \overline{g_1(\Y)} + \overline{f_2(\X)} + \overline{g_2(\Y)}$. Thus, for some functions $f,g \in \F_{2n}$, $\Z'=f(\overline{\X}) + g(\overline{\Y})$. Let $\overline{\X'} = f(\overline{\X})$ and $\overline{\Y'} = g(\overline{\Y})$. 

The following is the main result of this section. Assume that we have some random variables such that $\X$ and $\Y$ continue to be independent, and $H_{\infty}(\X), H_{\infty}(\Y) \ge k- 2n^{\delta}$.
\begin{lemma}\label{lem:new_acb}  There exists an efficiently computable function $\acb:\zo^{2n} \times\zo^{n_1} \rightarrow \zo^{m}$, $n_1= n^{\delta}$ and $m=n^{\Omega(1)}$, such that 
$$
\acb(\overline{\X} + \ol{\Y},w), \acb(\ol{f(\X)} + \ol{g(\Y)},w') \approx_{\epsilon} \U_m, \acb(\ol{f(\X)} + \ol{g(\Y)},w'),
$$
for any fixed strings $w,w' \in \zo^{n_1}$ with $w \neq w'$.
\end{lemma}
We use the rest of the section to prove the above lemma. In particular, we prove that the function $\acb$ computed by Algorithm $\ref{alg:ilnm}$ satisfies the conclusion of Lemma $\ref{lem:new_acb}$. 

We start by setting up some ingredients and parameters.
\begin{itemize}
\item Let $\delta>0$ be a small enough constant.
\item Let $n_2=n^{\delta_1}$, where $\delta_1 = 2 \delta$.
\item Let $\LExt_1: \zo^{n_2} \times \zo^{d} \rightarrow \zo^{d_1}$, $d_1= \sqrt{n_2}$,  be a linear-seeded extractor instantiated from Theorem $\ref{trev_ext}$ set to extract from entropy $k_1=n_2/10$ with error $\epsilon_1=1/10$. Thus $d=  C_1\log n_2$, for some constant $C_1$. Let $D=2^{d}=n^{\delta_2}$, $\delta_2=2C_1 \delta$.
\item Set $\delta' = 20  C_1 \delta$.
\item Let $\LExt_2: \zo^{2n} \times \zo^{d_1} \rightarrow \zo^{n_4}$, $n_4=n^{8\delta_3}$ be a linear-seeded extractor instantiated from Theorem $\ref{trev_ext}$ set to extract from entropy $k_2=0.9 k$ with error $\epsilon_2=2^{-\Omega(\sqrt{d_1})}=2^{-n^{\Omega(1)}}$, such that the seed length of the extractor $\LExt_2$ (by Theorem $\ref{trev_ext}$) is $d_1$. 
\item Let $\acb':\zo^{n_{1,acb'}} \times \zo^{n_{acb'}} \times \zo^{h_{acb'}} \rightarrow \zo^{n_{2,acb'}}$,  be the advice correlation breaker from Theorem $\ref{thm:acb}$ set with the following parameters:  $n_{acb'}=2n, n_{1,acb'}=n_4,n_{2,acb'} =m=O(n^{2\delta_2}), t_{acb'} = 2D, h_{acb'}=n_1+d, \epsilon_{acb'}= 2^{-n^{\delta}}$,  $d_{acb'}=O(\log^2(n/\epsilon_{acb'})),  \la_{acb'}=0$. It can be checked that by our choice of parameters, the conditions required for Theorem $\ref{thm:acb}$ indeed hold for $k_{1,acb'} \ge n^{2\delta_2}$.
 \end{itemize}

\RestyleAlgo{boxruled}
\LinesNumbered
\begin{algorithm}[ht]\label{alg:ilnm}
  \caption{$\acb(z)$  \vspace{0.1cm}\newline \textbf{Input:} 
  Bit-strings $z=(x \circ y)_{\pi}$ of length $2n$  and bit string $w$ of length $n_1$, where $x$ and $y$ are each $n$ bit-strings and $\pi:[2n] \rightarrow  [2n]$ is a permutation. \newline \textbf{Output:} Bit string  of length $m$. 
 }
Let $z_1= \slice(z,n_2)$.

Let $v$ be a $D \times n_3$ matrix, with its $i$'th row $v_i = \LExt_1(z_1,i)$.

Let $r$ be a $D \times n_4$ matrix, with its $i$'th row $r_i = \LExt_2(z,v_i)$.

Let $s$ be a $D \times m$ matrix, with its $i$'th row $s_i = \acb'(r_i,z, w \circ i)$.

Output $\oplus_{i=1}^D s_i$.
\end{algorithm}

Let $\X_1$ be the bits of $\X$ in $\Z_1$ and $\X_2$ be the remaining bit of $\X$.  Define $\Y_1$ and $\Y_2$ similarly. Without loss of generality suppose that $|\X_1| \ge |\Y_1|$. Let $\overline{\X}_1= \slice(\overline{\X},n_2)$ and $\overline{\Y}_1= \slice(\overline{\Y},n_2)$. Define $\overline{\X}_1' = \slice(f(\overline{\X}),n_2)$ and $\overline{\Y_1}' = \slice(g(\overline{\Y}),n_2)$. It follows that $\Z_1 = \overline{\X}_1 + \overline{\Y}_1$ and $\Z_1' = \overline{\X}_1' + \overline{\Y}_1'$.
\begin{claim}\label{cl:ilnm1} Conditioned on the random variables $\Y_1,\overline{\Y}_1'$, $\{ \LExt_2(\ox,\LExt_1(\ox_1+\oy_1,i))\}_{i=1}^{D}$, $\{ \LExt_2(\ox',\LExt_1(\ox'_1+ \oy'_1,i))\}_{i \in [D]}$, $\X_1$ and $\overline{\X}_1'$, the following hold:
\begin{itemize}
\item the matrix $\rr$ is $2^{-n^{\Omega(1)}}$-close to a somewhere random source,
\item $\rr$ and $\rr'$ are  deterministic functions of $\Y$,
\item $H_{\infty}(\X) \ge n-n^{\delta'}$, $H_{\infty}(\Y) \ge n- n^{\delta'}$.
\end{itemize}
\end{claim}
\begin{proof} 
By construction, we have that for any $j \in [D]$,
\begin{align*}
\rr_j &= \LExt_2(\Z,\LExt_1(\Z_1,j))   \\
       &=\LExt_2(\overline{\X}+ \overline{\Y}, \LExt_1(\overline{\X_1}+\overline{\Y_1},j))  \\
       &= \LExt_2(\overline{\X}, \LExt_1(\overline{\X}_1+\overline{\Y}_1,j)) + \LExt_2(\overline{\Y}, \LExt_1(\overline{\X}_1+\overline{\Y}_1,j))
\end{align*}
Similarly,
\begin{align*}
\rr_j' =  \LExt_2(\overline{\X}', \LExt_1(\overline{\X}_1'+\overline{\Y}_1',j)) + \LExt_2(\overline{\Y}', \LExt_1(\overline{\X}_1'+\overline{\Y}_1',j)).
\end{align*}
Fix the random variables $\Y_1,\oy_1'$. Note that after these fixings, $\oy$ has min-entropy at least $k- 2n_1 - n_2 >0.9k$. Now, since $\LExt_2$ is a strong seeded extractor for entropy $0.9k$, it follows that there exists a set $T \subset \zo^{d_1}$, $|T| \ge (1-\sqrt{\epsilon_2})2^{d_1}$, such that for any $j \in [T]$, $| \LExt_2(\oy,j) - \U_{n_4}| \le \sqrt{\epsilon_2}$. 

Now viewing $\LExt_1$ as a sampler (see Section $\ref{sec:samp_weak}$) using the weak source $\ox_{1,y_1}=\ox_1 + \overline{y_1}$, it follows by Theorem $\ref{thm:seed_samp}$ that 
$$      \Pr[ |\{\LExt_1(\ox_{1,y_1},i):i \in \zo^{d}\} \cap T| >(1-\sqrt{\epsilon_2} - \epsilon_1) D   ] \ge 1-2^{0.2 n_2} = 1-2^{-n^{\Omega(1)}}. $$

We fix $\ox_1$, and it  follows that with probability at least $1-2^{-n^{\Omega(1)}}$,  $\{\LExt_1(\ox_{1,y_1},i):i \in \zo^{d}\} \cap T \neq \emptyset$, and thus there exists a $j \in [D]$ such that  $ \LExt_2(\overline{\Y}, \LExt_1(\overline{\X_1}+\overline{\Y_1},j))$ is $2^{-n^{\Omega(1)}}$-close to $\U_{n_2}$ and is a deterministic function of $\Y$. 

We now fix the random variables $\ox_1'$, $\{\LExt_2(\overline{\X}, \LExt_1(\overline{\X_1} + \overline{\Y_1},i))\}_{i=1}^D$, $\{\LExt_2(\overline{\X}', \LExt_1(\overline{\X_1}' +\overline{\Y_1}',i))\}_{i=1}^D$, and note that $ \LExt_2(\overline{\Y}, \LExt_1(\overline{\X_1}+\overline{\Y_1},j))$ continues to be $2^{-n^{\Omega(1)}}$-close to $\U_{n_2}$. It follows that $\rr_j$ is  $2^{-n^{\Omega(1)}}$-close to $\U_{n_2}$. Further, for any $i \in [D]$, the random variables $\rr_i$ and $\rr_i'$ are deterministic functions of $\Y$. Finally, note that $\X$ and $\Y$ remain independent after these conditionings, and $H_{\infty}(\X) \ge n-3n_1-2n_2-2Dn_4 \ge n-n^{10\delta_2} $ and $H_{\infty}(\Y) \ge n-3n_1-n_2>n-n^{\delta_2}$.
\end{proof}
Lemma $\ref{lem:new_acb}$  is now direct from the next claim.
\begin{claim}\label{cl:ilnm2} There exists $j \in [D]$ such that 
$$ \s_j, \{ \s_i \}_{i \in  [D]\setminus j} \approx_{2^{-n^{\Omega(1)}}} \U_{m}, \{ \s_i \}_{i \in  [D]\setminus j}.$$
\end{claim}
\begin{proof} Fix the random variables:  $\W,\W',\Y_1,\oy_1'$,  $\{ \LExt_2(\ox,\LExt_1(\ox_1+\oy_1,i))\}_{i=1}^{D}$, $\{ \LExt_2(\ox',\LExt_1(\ox'_1+ \oy'_1,i))\}_{i \in [D]}$, $\X_1$ and $\ox_1'$. By Lemma $\ref{lem:new_adv_1}$, we have that with probability at least $1-2^{-n^{\Omega(1)}}$, $\W \neq \W'$. Further, by Claim $\ref{cl:ilnm1}$ we have that $\rr$ and $\rr'$ are deterministic functions of $\Y$, and with probability at least $1-2^{-n^{\Omega(1)}}$,  there exists $j \in [D]$ such that $\rr_j$ is $2^{-n^{\Omega(1)}}$-close to uniform, and  $H_{\infty}(\ox) \ge \frac{1}{2}n_{acb} - n^{\delta'}>n^{2\delta_2}$. Recall  that $\Z=\ox + \oy$ and $\Z' =\ox' + \oy'$. It now follows by Theorem $\ref{thm:acb}$ that 
\begin{align*}
\acb'(\rr_{j}, \Z, \W \circ j), \{ \acb'(\rr_{i},\ox+\oy, \W \circ i) \}_{i \in [D] \setminus j},  \{ \acb'(\rr_{i}',\ox'+\oy', \W' \circ i) \}_{i \in [D]} \approx_{2^{-n^{\Omega(1)}}} \\ \U_{m}, \{ \acb'(\rr_{i},\ox+\oy, \W \circ i) \}_{i \in [D] \setminus j}, \{ \acb'(\rr_{i}',\ox'+\oy', \W' \circ i) \}_{i \in [D]}
\end{align*} 
This completes the proof of the claim.
\end{proof}
\subsection{The non-malleable extractor}
\label{sec:new_ext_composede}
We are now ready to present the construction of $\ilnm$ that satisfies the requirements of Theorem~$\ref{theorem:ext_lin_composed_ss_2}$. 
\begin{itemize}
\item Let $\delta>0$ be a small enough constant, $n_1= n^{\delta}$ and $m=n^{\Omega(1)}$.
\item Let $\adv:\zo^{2n} \rightarrow \zo^{n_1}$, $n_1=n^{\delta}$, be the advice generator from Lemma $\ref{lem:new_adv_1}$.
\item Let $\acb:\zo^{2n} \times\zo^{n_1} \rightarrow \zo^{m}$ be the advice correlation breaker from Lemma $\ref{lem:new_acb}$.
\end{itemize}
\RestyleAlgo{boxruled}
\LinesNumbered
\begin{algorithm}[ht]\label{alg:ilnme}
  \caption{$\ilnm(z)$  \vspace{0.1cm}\newline \textbf{Input:} 
  Bit-string $z=(x \circ y)_{\pi}$ of length $2n$, where $x$ and $y$ are each $n$ bit-strings, and $\pi:[2n] \rightarrow  [2n]$ is a permutation. \newline \textbf{Output:} Bit string  of length $m$. 
 }
Let $w = \adv(z)$.

Output $\acb(z,w)$
\end{algorithm}
We prove that the function $\ilnm$ computed by Algorithm $\ref{alg:ilnme}$ satisfies the conclusion of Theorem $\ref{theorem:ext_lin_composed_ss_2}$ as follows.
Fix the random variables $\W, \W'$. By Lemma $\ref{lem:new_adv_1}$, it follows that $ \X$ remains independent of $ \Y$, and with probability at least $1-2^{-n^{\Omega(1)}}$, $H_{\infty}(\X) \ge k-2n_1$ and $H_{\infty}(\Y) \ge k-2n_1$ (recall $k \ge n-n^{\delta}$). Theorem $\ref{theorem:ext_lin_composed_ss_2}$ is now direct using  Lemma $\ref{lem:new_acb}$.

\section{Non-malleable extractors for split-state adversaries with bounded communication}
\label{section:communication_split_state}
Let $\F_{n,t} \subset \F_{2n}$ be the set of all functions that can be computed in the following way. Let $c= (x,y)$ be the input in $\zo^{2n}$, where $x$ is the first $n$ bits of $c$ and $y$ is the remaining $n$ bits of $c$. Let Alice and Bob be two tampering adversaries, where Alice has access to $x$ and Bob has access to $y$. Alice and Bob run a (deterministic) communication protocol based on $x$ and $y$ respectively, which can last for an arbitrary number of rounds but each party sends at most $t$ bits. Finally, based on the transcript and $x$ Alice outputs $x' \in \zo^n$, similarly based on the transcript and $y$ Bob outputs $y' \in \zo^n$. The function outputs $c'=(x',y')$. The following is our main result.
\begin{thm}
\label{thm:comm_nm_ext} There exists a constant $\delta>0$ such that  for all integers $n,t>0$ with $t \le \delta n$, there exists an efficiently computable function $\nmExt:\zo^{n} \times \zo^{n} \rightarrow \zo^{m}$, $m = \Omega(n)$, such that the following holds: let $\X$ and $\Y$ be uniform independent sources each on $n$ bits, and let $h$ be an arbitrary tampering function in $\F_{n,t}$. Then, there exists a distribution $\D_{h}$ on $\zo^{m} \cup \{ \same\}$ that is independent of $\X$ and $\Y$ such that $$|\nmExt(\X,\Y), \nmExt(h(\X,\Y)) - \U_m, \cpy(\D_{h},\U_m)| \le 2^{-\Omega(n\log \log n/\log n)}.$$
Further, $\nmExt$ is $2^{-\Omega(n \log \log n/\log n)}$-invertible.
\end{thm}
\begin{proof}
We show that any $2$-source non-malleable extractor that works for min-entropy $n-2\delta n$ can be used as the required non-malleable extractor in the above theorem. The  tampering function $h$ that is based on the communication protocol can be rephrased in terms of functions in the following way. Suppose the protocol lasts for $\ell$ rounds, there exist deterministic functions  $f_{i}$ and $g_{i}$ for $i =1,\ldots,\ell$, and $f:\zo^{n} \times \zo^{2 t} \rightarrow \zo^n$ and $g:\zo^{n} \times \zo^{2 t} \rightarrow \zo^n$ such that the  communication protocol between Alice and Bob corresponds to computing the following random variables: $\s_1 = f_1(\X), \rr_1 = g_1(\Y,\s_1),\s_2 = f_2(\X,\s_1,\rr_1),\ldots, \s_i = f_i(\X,\s_1,\ldots,\s_{i-1}, \rr_{1},\ldots,\rr_{i-1}), \rr_i = g_i(\Y,\s_1,\ldots,\s_{i},\rr_{i,},\ldots,\rr_{i-1}),\ldots,\rr_{\ell} = g_{\ell}(\Y,\s_1,\ldots,\s_{\ell},\rr_{1},\ldots,\rr_{\ell-1})$. 

Finally, $\X' = f(\X,\rr_1,\ldots,\rr_{\ell},\s_1,\ldots,\s_{\ell})$ and $\Y' = g(\Y,\rr_1,\ldots,\rr_{\ell},\s_1,\ldots,\s_{\ell})$ correspond to the output of Alice and  the output of Bob respectively. Thus,  $h(\X,\Y) = (\X',\Y')$.

Similar to the way we argue about alternating extraction protocols, we fix random variables in the following order: Fix $\s_1$, and it follows that $\rr_1$ is now a deterministic function of $\Y$. We fix $\rr_1$, and thus $\s_2$ is now a deterministic function of $\X$. Thus, continuing in this way, we can fix all the random variables $\s_1,\ldots,\s_{\ell}$ and $\rr_1,\ldots,\rr_{\ell}$ while maintaining that $\X$ and $\Y$ are independent. Further, invoking Lemma $\ref{lemma:entropy_loss_1}$, with probability at least $1-2^{-\Omega(n)}$, both $\X$ and $\Y$ have min-entropy at least $n-t - \delta n \ge n - 2 \delta n$ since both parties send at most $t$ bits. 

Note that now, $\X' = \eta(\X)$ for some deterministic function $\eta$ and $\Y' =\nu(\X)$ for some deterministic function $\nu$. Thus, for any $2$-source non-malleable extractor $\nmExt$ that works for min-entropy $n-2\delta n$ with error $\epsilon$, we have that there exists  a distribution $\D_{\eta,\nu}$ over $\zo^{m} \cup \{\same\}$ that is independent of $\X$ and $\Y$ such that $$|\nmExt(\X,\Y), \nmExt(\eta(\X), \nu(\Y)) - \U_m, \cpy(\D_{\eta,\nu},\U_m)| \le \epsilon+.2^{-\Omega(n)}$$
The theorem now follows by plugging in such a construction from a recent work of Li (\cite{Li18}, Theorem $1.12$). We note the non-malleable extractor in \cite{Li18} is indeed $2^{-\Omega(n \log \log n/\log n)}$-invertible.
\end{proof}
\section{Efficient sampling algorithms}
\label{sec:sampling}
In this section, we provide efficient sampling algorithms for the seedless non-malleable extractor construction presented in Section $\ref{section:composed_split_state}$. This is crucial to get efficient encoding algorithms for the corresponding non-malleable codes. We do not know how to invert the non-malleable extractor constructions in Theorem $\ref{theorem:ext_lin_composed_ss_1}$, but we show that the constructions can suitably modified in a way that admits efficient sampling from the pre-image of the extractor. 

\subsection{An invertible non-malleable extractor with respect to linear composed with interleaved adversaries}
The main idea is to ensure that on fixing appropriate random variables that are generated in computing the non-malleable extractor, the source is now restricted onto a known subspace of fixed dimension (i.e., the dimension does not depend on value of the fixed random variables). Once we can ensure this, sampling from the pre-image can simply be done by first uniformly sampling the fixed random variables, and then sampling the other variables uniformly from the known subspace. To carry this out, we need an efficient construction of a linear seeded extractor that has the property that for any fixing of the seed the linear map corresponding linear seeded extractor has the same rank.  Such a linear seeded extractor was constructed in prior works \cite{CGL15,Li16} (see Theorem $\ref{thm:low_error_inv_lin}$).

One additional care we need to take is the choice of the error correcting code we use in the advice generator construction. We ensure that the linear constraints imposed by fixing the advice string does not depend on the value of the advice string. This is  subtle  since the advice generator  comprises of a sample from an error correction of the sources as well as the output of  a linear seeded extractor on the source. The basic idea is to remove a few sampled coordinates of the error corrected sources and show that this suffices to remove any linear dependencies. 

We use the following notation: For any linear map  $L:\zo^{r} \rightarrow \zo^s$   given by $L(\alpha)=M \alpha$ for some matrix $M$, we use $con_{L}$ to denote a maximal set of linearly independent rows of $M$.  

We now set up some parameters and ingredients for our construction of an invertible non-malleable extractor. 
\begin{itemize}
\item Let $\delta>0$ be a small enough constant and $C$ a large constant. 
\item Let  $\delta'=\delta/C$.
\item Let $\C$ be a $\bch$ code with parameters: $[n_{b},n_{b}-t_{b} \log n_{b},2t_{b}]_2$, $t_{b}= \sqrt{n_b}/100$, where we fix  $n_{b}$ in the following way. Let $\dbch$ be the dual code.  From standard literature, it follows that $\dbch$ is a $[n_{b}, t_{b} \log n_{b}, \frac{n
_b}{2}- t_{b}\sqrt{n_{b}}]_2$-code.  Set $n_{b}$ such that $t_{b} \cdot \log n_{b}= \sqrt{n_b} \log n_b = 2n$. Let $E$ be the encoder of $\dbch$. Note that by our choice of parameters, the relative minimum distance of $\dbch$ is at least $1/3$.
%\item Let $\bb{F}$ be the finite field  $\bb{F}_{2^{\log (n+1)}}$. Let $n_7 = (2n - n_1 )/\log (n+1)$. Let $\RS: \bb{F}^{n_4} \rightarrow \bb{F}^{n}$ be the Reed-Solomon code encoding $n_7$ symbols of $\bb{F}$ to $n$ symbols in  $\bb{F}$, where we   use  $\RS$ to denote  the code as well as the encoder. Thus, $\RS$ is a $[n,n_7,n-n_7+1]_{n}$ error correcting code.
\item Let $n_0 = n^{\delta'}, n_1=  n_0^{c_0},   n_2= 10n_0$, for some  constant $c_0$ that we set below.
 \item Let $n_3= n^{C \delta}, n_4 = n^{C^2 \delta}/5, n_5 =n^{C^{3} \delta}, n_6 = n - \sum_{i=1}^5n_i $. We ensure that $n_6 = n(2-o(1))$.
\item Let $\Ext_1: \zo^{n_1} \times \zo^{d_1} \rightarrow \zo^{\log(n_b)}$ be a $(n_1/20,1/10)$-seeded extractor instantiated using Theorem $\ref{guv}$. Thus $d_1= c_1 \log n_1$, for some constant $c_1$. Let $D_1=2^{d_1}=n_1^{c_1}$.
\item Let $\samp_1:\zo^{n_1} \rightarrow [n_b]^{D_1}$ be the sampler obtained from Theorem $\ref{thm:seed_samp}$ using $\Ext_1$. 
\item Let $\Ext_2: \zo^{n_2} \times \zo^{d_2} \rightarrow \zo^{\log(n_6)}$ be a $(n_2/20,1/n_0)$-seeded extractor instantiated using Theorem $\ref{guv}$. Thus $d_2= c_2 \log n_2$, for some constant $c_2$. Let $D_2=2^d_2$.  Thus $D_2= 2^{d_2}= n_2^{c_2}$.
\item Let $\samp_2:\zo^{n_2} \rightarrow [n_6]^{D_2}$ be the sampler obtained from Theorem $\ref{thm:seed_samp}$ using $\Ext_2$.
\item Set $c_0 = 2c_2$.
\item Let $\ilext:\zo^{D_2} \rightarrow \zo^{n_0}$ be the extractor from Theorem $\ref{thm:il_ext}$.
\item Let $\LExt_0: \zo^{2n} \times \zo^{n_0} \rightarrow \zo^{\sqrt{n_0}}$ be a linear seeded extractor instantiated from Theorem $\ref{thm:strong_ip}$ set to extract from min-entropy $n_1/100$ and error $2^{-\Omega(\sqrt{n_0})}$.
\item Let $\Ext_3: \zo^{n_3} \times \zo^{d_3} \rightarrow \zo^{\log (n_6-D_2)}$ be a $(n_3/8,1/100)$-seeded extractor instantiated using Theorem $\ref{guv}$. Thus $d_3= C_1 \log n_3$, for some constant $C_1$.
\item Let $\samp_3:\zo^{n_3} \rightarrow [n_6-D_2]^{n_7}$ be the sampler obtained from Theorem $\ref{thm:seed_samp}$ using $\Ext_3$. Thus $n_7= 2^{d_3}= n_3^{C_1}$.
\item Let $\Ext_4: \zo^{n_4} \times \zo^{d_4} \rightarrow \zo^{n_6 - n_7-D_2}$ be a $(n_4/8,1/100)$-seeded extractor instantiated using Theorem $\ref{guv}$. Thus $d_3= C_1 \log n_4$. 
\item Let $\samp_4:\zo^{n_4} \rightarrow [n_5-n_7-D_2]^{n_{8}}$ be the sampler obtained from Theorem $\ref{thm:seed_samp}$ using $\Ext_4$. Thus $n_{8}= 2^{d_3}= n_4^{C_1 }$. 
\item Let $\LExt_1: \zo^{n_5} \times \zo^{d} \rightarrow \zo^{d_5}$, $d_5= \sqrt{n_5}$,  be a linear-seeded extractor instantiated from Theorem $\ref{trev_ext}$ set to extract from entropy $k_1=n_2/10$ with error $\epsilon_1=1/10$. Thus $d=  C_2\log n_5$, for some constant $C_2$. Let $D=2^{d}$.
\item Let $\LExt_2: \zo^{n_7} \times \zo^{d_5} \rightarrow \zo^{m_1}$, $m_1=\sqrt{n_7}$ be a linear-seeded extractor instantiated from Theorem $\ref{trev_ext}$ set to extract from entropy $k_2= n_7/100$ with error $\epsilon_2=2^{-\Omega(\sqrt{d_4})}=2^{-n^{\Omega(1)}}$, such that the seed length of the extractor $\LExt_2$ (by Theorem $\ref{trev_ext}$) is $d_5$. 
\item Let $\acb:\zo^{n_{1,acb}} \times \zo^{n_{acb}} \times \zo^{h_{acb}} \rightarrow \zo^{n_{2,acb}}$,  be the advice correlation breaker from Theorem $\ref{thm:acb}$ set with the following parameters:  $n_{acb}=n_7, n_{1,acb}=m_1,n_{2,acb} =n_{9}=D^2, t_{acb} = 2D, h_{acb}=n^{\delta}+d, \epsilon_{acb}= 2^{-n^{\delta'}}$,  $d_{acb}=O(\log^2(n/\epsilon_{acb})),  \la_{acb}=0$. It can be checked that by our choice of parameters, the conditions required for Theorem $\ref{thm:acb}$ indeed hold for $k_{1,acb} \ge n^{C\delta}$.
\item Let $\LExt_3: \zo^{n_{8}} \times \zo^{n_{9}} \rightarrow \zo^{m}$ be the linear seeded extractor from Theorem $\ref{thm:low_error_inv_lin}$ set to extract from min-entropy rate $1/10$ and error $\epsilon= 2^{-n^{\Omega(1)}}$ (such that the seed-length is indeed $n_9$). Thus, $m = \alpha n_{9}$, for some small contant $\alpha$ that arises out of Theorem $\ref{thm:low_error_inv_lin}$.
 \end{itemize}

\RestyleAlgo{boxruled}
\LinesNumbered
\begin{algorithm}[ht]\label{alg:inv_ilnm}
  \caption{$\ilnm(z)$  \vspace{0.1cm}\newline \textbf{Input:} 
  Bit-string $z=(x \circ y)_{\pi}$ of length $2n$, where $x$ and $y$ are each $n$ bit-strings, and $\pi:[2n] \rightarrow  [2n]$ is a permutation. \newline \textbf{Output:} Bit string  of length $m$. 
 }
Let $z_i = z_1 \circ z_2 \circ z_3 \circ z_4 \circ z_5 \circ z_6$, where $z_i$ is of length $n_i$.

Let $T_i= \samp_i(z_i)$, $i=1,2,3,4$.

Let  $\overline{z}_2=(z_6)_{T_2}$.

Let $z_2'=\ilext(\ol{z}_2)$.

Let $z_2''=\LExt_0(z,z_2')$.

For any set $Q \subseteq [2n]$, define  the linear function $E: \zo^{2n} \rightarrow \zo^{|Q|}$ as $E_Q(x) = (E(x))_Q$.

Pick a subset $\ol{T_1} \subset T_1$  of size $D_1-\sqrt{n_0}$ such that $con_{E_{\ol{T_1}}}$ is linearly independent of $con_{LExt_0(\cdot,z_2')}$. If there is no such set $\ol{T_1}$, then output $0^m$.

Let $w= z_1 \circ z_2 \circ \ol{z}_2 \circ (E( z))_{\ol{T_1}} \circ z_2''$.

Let $v$ be a $D \times d_4$ matrix, with its $i$'th row $v_i = \LExt_1(z_5,i)$.

Let $z_6'$ be the bits in $z_6$ outside $T_2$. Let $\ol{z_6}=(z_6')_{T_3}$.

Let $r$ be a $D \times n_4$ matrix, with its $i$'th row $r_i = \LExt_2(\ol{z_6},v_i)$. 

Let $s$ be a $D \times m$ matrix, with its $i$'th row $s_i = \acb(r_i,\ol{z_6}, w \circ i)$.

Let $\tilde{s} = \oplus_{i=1}^D s_i$.

Let $z_7$ be the bits in $z_6$ outside the coordinates $T_2 \cup T_3$.

Let $\ol{z_7} = (z_7)_{T_4}$. Let $z_8$ be the bits in $z_{6}$ outside the coordinates $T_2 \cup T_3 \cup T_4$.

Output $g = \LExt_3(\ol{z_7},\tilde{s})$.
\end{algorithm}

\begin{thm}
For all integers $n>0$  there exists an explicit function $\nmExt: \zo^{2n} \rightarrow \zo^{m}$,   $m=n^{\Omega(1)}$, such that the following holds: For any linear function $h : \zo^{2n} \rightarrow \zo^{2n}$, arbitrary tampering functions $f,g \in \F_n$,  any permutation $\pi:[2n] \rightarrow [2n]$ and  independent  uniform sources $\X$ and $\Y$ each on $n$ bits, there exists a distribution $\D_{h,f,g,\pi}$ on $\zo^m \cup  \{ \same\}$, such that 
$$ |\nmExt((\X \circ \Y)_{\pi}),  \nmExt(h((f(\X) \circ g(\Y))_{\pi})) - \U_{m},\cpy(\D_{h,f,g,\pi},\U_m) | \le 2^{-n^{\Omega(1)}}.$$
\end{thm}

The proof that $\ilnm$ computed by Algorithm $\ref{alg:inv_ilnm}$ satisfies  Theorem $\ref{theorem:ext_lin_composed_ss_1}$ is very similar, and we omit the details. We include a discussion of the key differences and subtleties that arise from the modifications done in the above construction  as compared to Algorithm $\ref{alg:ilnm}$. 

The first key difference is Step $7$, where we discard some bits from the advice generator's output. The existence of the subset $\ol{T}_1$ is guaranteed  by the fact that  $E$ has dual distance $t_b = \Omega(n/\log n)$. Thus, for any $T$, it must be that $\con_{E_{T_1}}$ is a set of size $|T_1| = D_1$. Further, $con_{LExt_0(\cdot,z_2')}$ is a set with cardinality at most $\sqrt{n_0}$.  Thus, indeed there exists such a set $\ol{T}_1$. An important detail to notice is that  $|T_1 \setminus \ol{T_1}| =  o(D_1)$ and the distance of the code computed by $E$ is $\Omega(1)$. Thus, the fact that we discard the bits indexed by the set $T_1 \setminus \ol{T_1}$ from the string $E(\Z)_{T_1}$ (and thus from the output of the advice generator) does not affect the correctness of the advice generator.

Another difference is that in the steps where we transform the somewhere random matrix $v$ into a matrix with longer rows, and the subsequent step where the advice correlation breaker is applied is now done using a pseudorandomly sampled subset of coordinates from $\Z$ (as opposed to  the entire $\Z$ which we did before). It is not hard to prove that this does not make a difference as long as we sample enough bits. Finally,  another difference is the final step where we use a linear seeded extractor, with $\ol{\Z_6}$ as the seed. As done many times in the paper, we use the sum structure of $\ol{\Z_7}$ (into a source that depends on $\X$ and a source that depends on $\Y$) along with the fact that $\LExt_3$ is linear seeded to show that the output is close to uniform.

 We now focus on the problem of efficiently sampling from the pre-image of this extractor. The following lemma almost immediately implies a simple sampling algorithm.
\begin{lemma}
\label{lem:sampling_interleaved}
With probability $1-2^{-n^{\Omega(1)}}$ over the fixing of the variables $z_1,z_2,\ol{z_2},z_2'',z_3,z_4,z_5,\ol{z_6}, w$, and any $g \in \zo^m$, the set $\ilnm^{-1}(g)$ is a linear subspace of fixed dimension.
\end{lemma}
\begin{proof} 
Consider any fixing of $z_1,z_2,z_3,z_4$. Clearly, these fix the sets $T_i$, $i=1,2,3,4$. Next, note that given $\overline{z}_2$, we have the value of  $z_2'$.  We note that by  Lemma $\ref{aff_error}$ that with probability $1-2^{-n^{\Omega(1)}}$, the linear map $\LExt_0(,z_2')$ has full rank. Using Algorithm $\ref{alg:ilnm}$, determine the set $\ol{T_1}$ (if it exists). Fix $E(z)_{\overline{T}_1}$ and $z_2''$, noting that the value of $w$ is now determined. Now given $z_5,\ol{z_6}$, we can compute $r,s,\tilde{s}$.  Next observe that given $g$ and $\tilde{s}$, Theorem  $\ref{thm:low_error_inv_lin}$ implies  the value of $\ol{z_7}$ belongs to a subspace whose dimension does not depend on the values of $g$ and $\tilde{s}$. Finally, we are left to see how to compute $z_8$. Note that the constraints on $z_8$ are imposed by the fixings of $z_2''$ and $E(C)_{\ol{T_1}}$. However, by construction (Step $7$ of our algorithm), the number of independent linear constraints on $z_8$ is exactly equal to $D_1$ as long as $\LExt_0(,z_2')$ has full rank (which as noted before occurs with probability at least  $1-2^{-n^{\Omega(1)}}$).  This completes the proof.
\end{proof}
Given Lemma $\ref{lem:sampling_interleaved}$, the sampling algorithm is now straightforward: 

Input $g \in \zo^{m}$; Output $z$ that is uniform on the set $\ilnm^{-1}(g)$.
\begin{enumerate}
\item Sample $z_i$, $i=1,2,3,4,5$ uniformly at random. Compute $T_1,T_2,T_3,T_4$ following Algorithm~$\ref{alg:ilnm}$.
\item Sample $\ol{z_2}$ uniformly, and compute $z_2'$. Further, sample $z_2''$ uniformly.
\item Compute $\ol{T_1}$, and sample $(E(z))_{\ol{T_1}}$ uniformly at random.
\item Compute $w,v,r,s,\tilde{s}$ using Algorithm $\ref{alg:ilnm}$.
\item Sample $\ol{z_7}$ from $(\LExt_3(\cdot,\tilde{s}))^{-1}(g)$ efficiently using Theorem $\ref{thm:low_error_inv_lin}$.
\item Sample $z_8$ as described in Lemma $\ref{lem:sampling_interleaved}$. Compute the string $z_6$.
\item Output $z = z_1 \circ z_2 \circ z_3 \circ z_4 \circ z_5 \circ z_6$.
\end{enumerate}

\section{Extractors for interleaved sources}
\label{sec:ilext}
Our techniques yield improved explicit constructions of extractors for interleaved sources. Our extractor works when both sources have entropy at least $2n/3$, and outputs $\Omega(n)$ bits that are $2^{-n^{\Omega(1)}}$-close to uniform.  

The following is our main result.
\begin{thm}
\label{thm:il_ext}
For any constant $\delta>0$ and  all integers $n>0$, there exists an efficiently computable function $\ilext: \zo^{2n} \rightarrow \zo^{m}$, $m = \Omega(n)$, such that for any two independent sources $\X$ and $\Y$, each on $n$ bits with min-entropy at least $(2/3 + \delta)n$, and any permutation  $\pi:[2n] \rightarrow [2n]$, we have $$|\ilext((\X \circ \Y)_{\pi}) - \U_m| \le 2^{-n^{\Omega(1)}}.$$
\end{thm}

We use the rest of the section to prove Theorem $\ref{thm:il_ext}$. An important ingredient in our construction is an explicit somewhere condenser for high-entropy sources constructed in the works of Barak et al.\ \cite{BRSW12} and Zuckerman \cite{Zuck07}.
\begin{thm}
\label{thm:condense}
For all constants $\beta, \delta$ and all integers $n>0$, there exists an efficiently computable function $\con:\zo^{n} \times \zo^{d} \rightarrow \zo^{\ell}$, $d = 0(1)$ and $\ell = \Omega(n)$ such that  the following holds: for any $(n,\delta n)$-source $\X$ there exists a $y \in \zo^{d}$ such that $\con(\X,y)$ is $2^{-\Omega(n)}$-close to a source with min-entropy $(1- \beta)\ell$.   \\We call such a function $\con$ to be a $(\delta,1-\beta)$-condenser.
\end{thm}

We prove that Algorithm $\ref{alg:il_ext}$ computes the required extractor. We begin by setting up some ingredients and parameters.

\begin{itemize}
\item Let $\kappa>0$ be a small enough constant.
\item Let $n_1 = (2/3 + \delta/2)n$ and $n_2 = n^{5\kappa}$.
\item Let $\beta$ be a parameter which we fix later. Let $\con: \zo^{n_1} \times \zo^{d} \rightarrow \zo^{\ell}$ be a $(\delta/4,1-\beta)$-condenser instantiated from Theorem $\ref{thm:condense}$. Thus $\ell = n/C'$, for some constant $C'$ that depends on $\delta, \beta$. Let $D= 2^d$. Note that $D= O(1)$.
\item Let $\LExt_1: \zo^{2n} \times \zo^{\ell} \rightarrow \zo^{n_2}$ be the linear seeded extractor from Theorem $\ref{thm:low_error_inv_lin}$ set to extract from min-entropy rate $1/12$ and error $\epsilon_1= 2^{-2 \beta \ell}$. The seed-length is at most $3 C\beta \ell$, some  constant $C$ that arises out of Theorem $\ref{thm:low_error_inv_lin}$. We choose $\beta = min\{1/3C,\gamma\}$, where $\gamma$ is the constant in Theorem $\ref{thm:low_error_inv_lin}$. Note that the seed-length of $\LExt_1$ is indeed  at most $\ell$.  
\item Let $\acb:\zo^{n_{1,acb}} \times \zo^{n_{acb}} \times \zo^{h_{acb}} \rightarrow \zo^{n_{2,acb}}$,  be the advice correlation breaker from Theorem $\ref{thm:acb}$ set with the following parameters:  $n_{acb}=2n, n_{1,acb}=n_2, n_{2,acb} =n_3= n^{2\kappa}, t_{acb} = D, h_{acb}=d, \epsilon_{acb}= 2^{-n^{\kappa}}, d_{acb}= O(\log^2(n/\epsilon_{acb})),  \la_{acb}=0$. It can be checked that by our choice of parameters, the conditions required for Theorem $\ref{thm:acb}$ indeed hold for $k_{1,acb} \ge n^{2 \kappa}$.
\item Let $\LExt_2: \zo^{2n} \times \zo^{n_3} \rightarrow \zo^{m}$, $m=\Omega(n)$,  be a linear-seeded extractor instantiated from Theorem $\ref{trev_ext}$ set to extract from entropy $k_1=n/10$ with error $\epsilon_1=2^{-\alpha\sqrt{n_3}}$, for an appropriately picked small constant $\alpha$.
 \end{itemize}

\RestyleAlgo{boxruled}
\LinesNumbered
\begin{algorithm}[ht]\label{alg:il_ext}
  \caption{$\ilext(z)$  \vspace{0.1cm}\newline \textbf{Input:} 
  Bit-string $z=(x \circ y)_{\pi}$ of length $2n$, where $x$ and $y$ are each $n$ bit-strings, and $\pi:[2n] \rightarrow  [2n]$ is a permutation. \newline \textbf{Output:} Bit string  of length $m$. 
 }
Let $z_1= \slice(z,n_1)$.

Let $v$ be a $D \times n_2$ matrix, with its $i$'th row $v_i = \con(z_1,i)$.

Let $r$ be a $D \times n_3$ matrix, with its $i$'th row $r_i = \LExt_1(z,v_i)$.

Let $s$ be a $D \times m$ matrix, with its $i$'th row $s_i = \acb(r_i,z,  i)$.

Let $\tilde{s} = \oplus_{i=1}^D s_i$.

Output $\LExt_2(z,\tilde{s})$.
\end{algorithm}
We use the following notation: Let $\X_1$ be the bits of $\X$ in $\Z_1$ and $\X_2$ be the remaining bit of $\X$. Let  $\Y_1$ be the bits of $\Y$ in $\Z_1$ and $\Y_2$ be the remaining bits of $\Y$. Without loss of generality assume $|\X_1| \ge |\Y_1|$. Define  $\overline{\X}= (\X \circ 0^n)_{\pi}$ and $\overline{\Y} = (\Y \circ 0^n)_{\pi}$. Further, let $\overline{\X}_1= \slice(\overline{\X},n_1)$ and $\overline{\Y}_1= \slice(\overline{\Y},n_1)$. It follows that $\Z= \overline{\X} + \overline{\Y}$, and $\Z_1 = \overline{\X_1} + \overline{\Y_1}$. Further, let $k_x =k_y = (2/3 + \delta)n$.

We begin by proving the following claim.
\begin{claim}\label{cl:il1} Conditioned on the random variables $\X_1,\Y_1, \{\LExt_1(\overline{\X}, \con(\overline{\X_1} + \overline{\Y_1},i))\}_{i=1}^D$, the following hold:
\begin{itemize}
\item the matrix $\rr$ is $2^{-\Omega(n)}$-close to a somewhere random source,
\item $\rr$ is a  deterministic functions of $\Y$,
\item $H_{\infty}(\X) \ge \delta n/4 $, $H_{\infty}(\Y) \ge n/6$.
\end{itemize}
\end{claim}
\begin{proof} 
By construction, we have that for any $j \in [D]$,
\begin{align*}
\rr_j &= \LExt_1(\Z,\con(\Z_1,j))   \\
       &=\LExt_1(\overline{\X}+ \overline{\Y}, \con(\overline{\X_1}+\overline{\Y_1},j))  \\
       &= \LExt_2(\overline{\X}, \con(\overline{\X_1}+\overline{\Y_1},j)) + \LExt_2(\overline{\Y}, \con(\overline{\X_1}+\overline{\Y_1},j))
\end{align*}
Fix the random variables $\Y_1$, and  $\oy$ has min-entropy at least $k_y- n_1/2\ge n/6 + 3 \delta n/4$. Further, note that $\ol{\X_1}$ has min-entropy at least $ n_1/2  - (n- k_x) \ge \delta n/4$. Now, by Theorem $\ref{thm:condense}$, we know that there exists a $j \in [D]$ such that $\con(\ol{\X}_1 + \ol{\Y}_1,j)$  is $2^{-\Omega(n)}$-close to a source with min-entropy at least $(1 - \beta)\ell$. Further, note that $\V$ is a deterministic function of $\X$. 

Now, since $\LExt_1$  is a strong seeded extractor set to extract from min-entropy $n/6$, it follows that  $$| \LExt_1(\oy,\con( \ol{\X}_1 + \ol{\Y}_1, j)) - \U_{n_2}| \le 2^{\beta \ell}\epsilon_1 + 2^{-\Omega(n)} \le 2^{-\beta \ell + 1}.$$

We now fix the random variables $\ol{\X}_1$ and note that $ \LExt_1(\oy,\con( \ol{\X}_1 + \ol{\Y}_1, j))$ continues to be $2^{-\Omega(\ell)}$-close to $\U_{n_2}$. This follows from the fact that $\LExt_1$ is a strong seeded extractor. Note that the random variables $\{\con(\ol{\X}_1 + \ol{\Y}_1, i)): i \in [D]\}$ are now fixed. Next, fix the random variables $\{\LExt_1(\overline{\X}, \con(\overline{\X_1} + \overline{\Y_1},i))\}_{i=1}^D$ noting that they are deterministic functions of $\X$. Thus $\rr_j$ is  $2^{-\Omega(n)}$-close to $\U_{n_2}$ and for any $i \in [D]$, the random variables $\rr_i$  are deterministic functions of $\Y$. Finally, note that $\X$ and $\Y$ remain independent after these conditionings, and $H_{\infty}(\X) \ge k_x -n_1 -D n_2 $ and $H_{\infty}(\Y) \ge k_y - n_1/2$.
\end{proof}
The next claim almost gets us to Theorem $\ref{thm:il_ext}$.
\begin{claim}\label{cl:il2} There exists $j \in [D]$ such that 
$$ \s_j, \{ \s_i \}_{i \in  [D]\setminus j}, \X \approx_{2^{-n^{\Omega(1)}}} \U_{n_3}, \{ \s_i \}_{i \in  [D]\setminus j}, \X.$$
\end{claim}
\begin{proof} Fix the random variables:  $\X_1,\Y_1, \{\LExt_1(\overline{\X}, \con(\overline{\X_1} + \overline{\Y_1},i))\}_{i=1}^D$.  By Claim $\ref{cl:il1}$ we have that $\rr$ is a deterministic function of $\Y$, and with probability at least $1-2^{-\Omega(n)}$,  there exists $j \in [D]$ such that $\rr_j$ is $2^{-n^{\Omega(1)}}$-close to uniform, and  $H_{\infty}(\ox) \ge \delta n/4 $. Recall  that $\Z=\ox + \oy$. It now follows by Theorem $\ref{thm:acb}$ that 
\begin{align*}
\acb(\rr_{j}, \Z, \W \circ j), \{ \acb(\rr_{i},\ox+\oy, \W \circ i) \}_{i \in [D] \setminus j}, \X \approx_{2^{-n^{\Omega(1)}}} \\ \U_{n_3}, \{ \acb(\rr_{i},\ox+\oy, \W \circ i) \}_{i \in [D] \setminus j}, \X.
\end{align*} 
\end{proof}
It follows by Claim $\ref{cl:il2}$ that $\widetilde{\s}$ is $2^{-n^{\Omega(1)}}$-close to uniform even conditioned on $\X$. Thus, noting that  $\LExt_2(\Z,\widetilde{\s}) = \LExt_2(\ol{\X},\widetilde{\s}) + \LExt_2(\ol{\Y},\widetilde{\s})$, it follows that we can fix $\widetilde{\s}$ and $\LExt_2(\ol{\X},\widetilde{\s})$ remains $2^{-n^{\Omega(1)}}$-close to uniform and is a deterministic function of $\X$. Next, we fix $\LExt_2(\ol{\Y},\widetilde{\s})$ without affecting the distribution of $\LExt_2(\ol{\X},\widetilde{\s})$. It follows that $\LExt_2(\Z,\widetilde{\s})$ is $2^{-n^{\Omega(1)}}$-close to uniform. This completes the proof of Theorem $\ref{thm:il_ext}$.

\bibliographystyle{alpha}
\bibliography{nm_interleaved}
 \end{document}